\documentclass[11pt,a4paper]{article}
\usepackage{amsfonts,amsmath,amssymb,amsthm}
\usepackage{latexsym,amscd}
\usepackage{amsbsy}
\usepackage{CJK}
\usepackage{indentfirst,mathrsfs}
\usepackage{latexsym,amscd}
\usepackage{amsbsy}
\usepackage[noadjust]{cite}
\usepackage{color}
\usepackage{multirow}
\usepackage{makecell}
\usepackage{float}
\usepackage{booktabs}
\usepackage[margin=2.5cm]{geometry}

\usepackage{scalerel,stackengine}
\stackMath
\newcommand\reallywidehat[1]{
	\savestack{\tmpbox}{\stretchto{
			\scaleto{
				\scalerel*[\widthof{\ensuremath{#1}}]{\kern-.6pt\bigwedge\kern-.6pt}
				{\rule[-\textheight/2]{1ex}{\textheight}}
			}{\textheight}
		}{0.5ex}}
	\stackon[1pt]{#1}{\tmpbox}
}

\numberwithin{equation}{section}
\newtheorem{theo}{Theorem}
\newtheorem{lem}{Lemma}
\newtheorem{coro}{Corollary}

\newtheorem{definition}{Definition}
\newtheorem{example}{Example}
\newtheorem{rem}{Remark}

\title{Multilevel inserting constructions for constant dimension subspace codes}
\author{Gang Wang\textsuperscript{1}\and Xuan Gao\textsuperscript{1} \and Sihem Mesnager\textsuperscript{2,3}\and Fang-Wei Fu\textsuperscript{4}}
\date{\small\textsuperscript{1}College of Science, Civil Aviation University of China, 300300, Tianjin, China. E-mail: gwang06080923@mail.nankai.edu.cn\\
	\small\textsuperscript{2}Department of Mathematics, University of Paris VIII, F-93526 Saint-Denis, Laboratory Geometry, Analysis and Applications, LAGA, University Sorbonne Paris Nord, CNRS, UMR 7539, F-93430, Villetaneuse, France.  E-mail: smesnager@univ-paris8.fr \\
	\small\textsuperscript{3}Telecom Paris, Polytechnic Institute of Paris, 91120 Palaiseau, France.\\
	\small\textsuperscript{4}Chern Institute of Mathematics and LPMC, Nankai University, 300071, Tianjin, China. E-mail: fwfu@nankai.edu.cn.\\}

\begin{document}
	
	\maketitle
	
	\begin{abstract}

		Subspace codes, especially constant dimension subspace codes (CDCs), represent an intriguing domain that can be used to conduct basic coding theory investigations. They have received widespread attention due to their applications in random network coding.  This paper presents inverse bilateral multilevel construction by introducing inverse bilateral identifying vectors and inverse bilateral Ferrers diagram rank-metric codes. By inserting the inverse bilateral multilevel construction into the double multilevel  construction and bilateral multilevel construction, an effective construction for CDCs is provided. Furthermore, via providing a new set of  bilateral identifying vectors, we give another efficient construction for CDCs. In this article, several CDCs are exhibited, equipped with the rank-metric, with larger sizes than the known ones in the existing literature. From a practical standpoint, our results could help in the pragmatic framework of constant-dimension-lifted rank-metric codes for applications in network coding.  The ratio of the new lower bound to the known upper bound for some CDCs is calculated, which is greater than 0.94548 for any prime power $q \geq 3.$
		
	\end{abstract}
	
	

	\section{Introduction}
	
	Throughout this paper, $q$ denotes a prime power and $\mathbb{F}_q$ is the finite field of order $q$.  Given a finite set $S$, its cardinality is denoted by $\vert S\vert$.  Let $\mathbb{F}_q^n$ denote the $n$-dimensional vector space over $\mathbb{F}_q$.  The set of all subspaces of  $\mathbb{F}_q^n$, called the projective space of order $n$ over  $\mathbb{F}_q$ denoted by $\mathcal{P}_q(n)$.  Given a nonnegative integer $k \leq n$,   the set of all $k$-dimensional 
	subspaces of  $\mathbb{F}_q^n$, forms the Grassmannian space (Grassmannian in short) over $\mathbb{F}_q$, which is denoted by $\mathcal{G}_q(n, k)$.

	It is well-known that, 
	$$\left|\mathcal{G}_q(n, k)\right|=\genfrac[]{0pt}{}{n}{k}_q \triangleq \prod_{i=0}^{k-1} \frac{q^{n-i}-1}{q^{k-i}-1},$$ 
	where $\genfrac[]{0pt}{}{n}{k}_q$ is $q$-ary Gaussian coefficient.
	
	A (subspace) code $\mathcal{C}$  is a nonempty collection of subspaces of $\mathbb{F}_q^n$, i.e., a nonempty
	subset of $\mathcal{P}_q(n)$.  We emphasise that, unlike classical coding theory in which each codeword is
	a vector, here, each codeword of $\mathcal{C}$ is itself an entire space of vectors. A code in
	which each codeword has the same dimension, i.e., a code contained within a
	single Grassmannian, is called a constant dimension code (CDC).
	
	As in classical coding theory, defining a distance measure between codewords is important. One possible distance measure between  two subspace $\mathcal{U}$, $\mathcal{V}$ of $\mathbb{F}_q^n$ is 
	$$
	d_S(\mathcal{U}, \mathcal{V})=\operatorname{dim}(\mathcal{U})+\operatorname{dim}(\mathcal{V})-2 \operatorname{dim}(\mathcal{U} \cap \mathcal{V}).
	$$

	Such subspace metric was initially introduced in the context of error- and erasure-correction in linear network coding (see \cite{16}).
	A subset $\mathcal{C}$ of $\mathcal{G}_q(n, k)$ is called an $(n, M, d, k)_q$-CDC, if $d_S(\mathcal{U}, \mathcal{V}) \geq d$ for all $\mathcal{U}$, $\mathcal{V} \in \mathcal{C}$ with $\mathcal{U} \neq \mathcal{V}$, and $\mathcal{C}$ has size $M$. Given $n, d, k$ and $q$, denote by $A_q(n, d, k)$ the maximum number of codewords among all $(n, d, k)_q$-CDCs. An $(n, d, k)_q$-CDC with $A_q(n, d, k)$ codewords is said to be optimal. One main problem for the CDCs is determining the maximum possible size $A_q(n, d, k)$ for given parameters $n$, $d$, $k$ and $q$.
	
	In recent years, several researchers have improved the construction of CDCs to increase the size of $A_q (n, d, k)$. K\"{o}tter and Kschischang \cite{16} introduced subspace codes in their research on random network coding and proposed that the lifted maximum rank distance (MRD) codes can produce asymptotically optimal CDCs. They also gave an upper bound of CDCs. Etzion and Silberstein \cite{4} provided the multilevel construction by introducing identifying vectors and Ferrers diagram rank-metric codes (FDRMCs), which generalized the lifted MRD codes. Trautman and Rosenthal \cite{28} introduced pending dots to improve the multilevel construction. Etzion and Silberstein in \cite{6}, and Silberstein and Trautman in \cite{27} constructed large subspace codes in $\mathcal{G}_q(n, k)$ of minimum subspace distance $d=4$ or $2k-2$ by pending dots and graph matching. In \cite{29}, Xu and Chen presented a construction which can also be seen as a generalization of the lifted MRD codes. Liu, Chang and Feng in \cite{18} generalized the construction of \cite{29} to parallel construction and showed parallel multilevel construction by combining the parallel and multilevel construction. In \cite{19}, Liu and Ji generalized parallel construction, named inverse multilevel construction, and pointed out the double multilevel construction by combining the inverse and multilevel construction. Recently, Yu and Ji \cite{30} put forward the bilateral multilevel construction, which generalized the multilevel construction.
	
	In addition, Heinlein \cite{11} presented new upper bounds of CDCs, which contain lifted MRD codes and improved the
	construction of  LMRD codes. 
	\begin{lem}(Theorem 1 in {\cite{11}})\label{lem1} For $n \geq 2k$, let $\mathcal{C}$ be an
		$(n, 2\delta, k)_q$-CDC which contains a lifted MRD code.\\
		(1) If $k < 2\delta$ and $n \geq 3\delta$, then
		$$
		\begin{gathered}
			\bar{A}_q(n, 2 \delta, k) \leq q^{(n-k)(k-\delta+1)} + {A}_q(n-k, 2(2\delta-k), \delta).
		\end{gathered}
		$$
		
		If additionally $k=\delta$, $n \equiv r~(mod~k)$, $0 \leq r < k$, and 
		$\genfrac[]{0pt}{}{r}{1}_q < k$, or $(n, 2\delta , k) \in \left\lbrace (6+3l, 4+2l, 3\notag \right . \\ \left.+l), (6l,4l,3l)~|~l\geq 1\right\rbrace$, 
	    	then the bound can be achieved.\\
		(2) If $k < 2\delta$ and $n < 3\delta$, then
		$$
		\begin{gathered}
			\bar{A}_q(n, 2 \delta, k) = q^{(n-k)(k-\delta+1)} + 1.
		\end{gathered}
		$$\\
		(3)If $2 \delta \leq k < 3\delta$, then
		$$
		\begin{gathered}
			\bar{A}_q(n, 2 \delta, k) \leq q^{(n-k)(k-\delta+1)} + {A}_q(n-k, 6\delta-2k), 2\delta)+q^{(k-2\delta +1)(n-k-\delta)} \frac{\genfrac[]{0pt}{}{n-k}{\delta}_q \genfrac[]{0pt}{}{k}{2\delta-1}_q}{\genfrac[]{0pt}{}{k-\delta}{\delta-1}_q}.
		\end{gathered}
		$$
		\end{lem}
		\noindent For more information on constructions and bounds for subspace code, the interested reader is referred to
	\cite{2,6,7,8,10,33,12,13,15,17,22,23,24,25,26}, the references therein. A helpful link for concrete subspace code is: http://subspacecodes.uni-bayreuth.de.
	
	This paper is based on the above constructions and is devoted to constructing large CDCs that contain a lifted MRD code as a subset. Multilevel construction plays an important role for CDCs, yet picking up the identifying vectors effectively is still an open and difficult problem. The inverse bilateral multilevel construction and the multilevel inserting construction help us reduce this difficulty. Applying the Theorem 2 and Theorem 3, we establish the new lower bounds for some CDCs (see Corollaries 2, 4, 5, 6). Some CDCs larger than the previously best-known codes are given (see Table 1).
	
	The structure of this paper is as follows. In Section 2, we introduce some basic definitions and previous results of CDCs.
	Section 3 presents the inverse bilateral multilevel construction by introducing inverse bilateral identifying vectors and inverse bilateral Ferrers diagram rank-metric codes (see Theorem 1).  Section 4 provides a new construction (see Theorem 2) by combining the double multilevel construction, the bilateral multilevel construction and the inverse bilateral multilevel construction. Moreover, by introducing a new set of bilateral identifying vectors, we propose another construction for CDCs. Finally, Section 5 concludes this contribution.

	\section{Preliminaries}
	
	This section will introduce some basic concepts, present some definitions,  fix our notation and discuss  some previous results concerning CDCs.
	
	A rank-metric code (RMC) $\mathcal{D}$ is a subspace of the matrix space $\mathbb{F}_q^{m \times n}$, and the rank distance on $\mathbb{F}_q^{m \times n}$ is defined by
	$$
	d_R(A, B)=\operatorname{rank}(A-B), \text { for } A, B \in \mathbb{F}_q^{m \times n} .
	$$
	An $(m \times n, M, \delta)_q$ RMC $\mathcal{D}$ is a subset of $\mathbb{F}_q^{m \times n}$ with cardinality $M$ and minimum rank distance
	$$
	\delta=\min _{A, B \in \mathcal{D}, A \neq B}\left\{d_R(A, B)\right\} .
	$$
	The cardinality of RMCs has a Singleton-like upper bound, that is $M \leq q^{\max \{m, n\}(\min \{m, n\}-\delta+1)}$. When the equality holds, we call $\mathcal{D}$ a maximum rank distance (MRD) code. Furthermore, if $\mathcal{D}$ is a $k$-dimensional $\mathbb{F}_q$-linear subspace of $\mathbb{F}_q^{m \times n}$, then it is said to be linear, and is denoted by an $[m \times n, k, \delta]_q$ RMC. Linear MRD codes exist for all feasible parameters (cf. \cite{3,9}).
	
We write ${I}_{k}$ to denote the $k \times k$ identity matrix.
	
	\begin{lem}{\cite{26}}\label{lem2} (\textbf{Lifted MRD codes}) Let $n \geq 2k$. The lifted MRD code
		$$
		\mathcal{C}=\left\{{rs}\left({I}_{k} \mid {A}\right): {A} \in \mathcal{D}\right\}
		$$
		is an $\left(n, q^{(n-k)(k-\delta+1)}, 2 \delta, k\right)_q$-CDC, where $\mathcal{D}$ is a $[k \times(n-k), \delta]_{q}$-MRD code.
	\end{lem}
	Given $n, k, \delta$ and $q$, denote by $\bar{A}_q(n, 2 \delta, k)$ the maximum number of codewords among all $(n, 2 \delta, k)_q$-CDCs containing a lifted MRD code as a subset. There are many known constructions of CDCs in the literature (cf. \cite{4,6,26,27,28}). Etzion and Silberstein in \cite{6} initially proposed the lower and upper bounds on determining $\bar{A}_q(n, 2 \delta, k)$.
	
	In $\mathbb{F}_q^{m \times n}$, the rank distribution of a RMC $\mathcal{C}$ is the $\min\{m,n\}+1$ vector whose $r$-th component is $a(q, m, n, \delta, r)(\mathcal{C})=|\{A \in \mathcal{C}: \operatorname{rank}(A)=r\}|$ with $0 \leq r \leq \min\{m,n\}$.
		
	\begin{lem}\cite{3, 9}\label{lem3} Let $\mathcal{C} \subset \mathbb{F}_q^{m \times n}$ be an MRD code with rank distance $\delta$. For $\delta \leq r \leq \min \{m, n\}$, its rank distribution is given by \\
		$$a(q, m, n, \delta, r)= \genfrac[]{0pt}{}{\min\{m,n\}}{r}_q \sum_{s=0}^{r-\delta}(-1)^s q^{\binom{s}{2}}
		\genfrac[]{0pt}{}{r}{s}_q
		\left(q^{\max \{m, n\}(r-s-\delta+1)}-1\right).$$
		
	\end{lem}
	
	Liu, Chang and Feng defined rank-metric codes with given ranks in \cite{18}, which is useful for the construction of CDCs.
	
	\begin{definition} \cite{18}\label{definition1} Let $K \subset\{0,1, \cdots, n\}$ and $\delta$ be a positive integer. We say $\mathcal{D} \subset \mathbb{F}_q^{m \times n}$ is an $(m \times n, M, \delta, K)_q$ rank-metric code with given ranks (GRMC) if it satisfies:\\
		(1) ${rank}({D}) \in K$ for any ${D} \in \mathcal{D}$;\\
		(2) $d_R\left({D}_1, {D}_2\right)={rank}\left({D}_1-{D}_2\right) \geq \delta$ for any ${D}_1, {D}_2 \in \mathcal{D}$, and ${D}_1 \neq {D}_2$;\\
		(3) $|\mathcal{D}|=M$.
	\end{definition}
	
	When $K=\left\lbrace 0,1,\cdot\cdot\cdot,n\right\rbrace $, a GRMC is simply an ordinary RMC that may not be linear. When $K = \left\lbrace t\right\rbrace$ for $0 \leq
	t \leq n$, a GRMC is called constant-rank code (cf. \cite{33}).
	Usually, $K$ is selected as a set of consecutive integers. Let $\left[t_1, t_2\right]$ denote the set of all integers $x$ such that $t_1 \leq x \leq t_2$.
	Given $m, n, K$ and $\delta$, denote by $A_q^G(m \times n, \delta, K)$ the maximum number of codewords among all $(m \times n, \delta, K)_q$-GRMCs. 
	 A lower bound for $A_q^G(m \times n, \delta, [t_1, t_2])$  is as follows.
	
	\begin{lem}\cite{18}\label{lem4} Let $1 \leq \delta \leq n \leq m$. Let $t_1$ be a nonnegative integer and $t_2$ be a positive integer such that $t_1 \leq t_2 \leq n$. Then
		
		 $A_q^G\left(m \times n, \delta,\left[t_1, t_2\right]\right) \geq$
	    $
		\begin{cases}\sum_{i=t_1}^{t_2} a(q, m, n, \delta, i), &   {if}\ t_2 \geq \delta ; \\ \max _{\max\{1, t_1 \leq d<\delta\}}\left\{\left[\frac{\sum_{i=\max \left\{1, t_1\right\}}^{t_2} a(q, m, n, d, i)}{q^{m(\delta-d)}-1}\right]\right\}, &  { otherwise. }\end{cases}
		$
	\end{lem}
	
	According to Definition 1, Liu, Chang and Feng in \cite{18} provided the parallel construction and it is a generalization of Lemma 2.
	\begin{lem}{\cite{18}}\label{lem5} (\textbf{Parallel  construction}). Let $n \geq 2 k \geq 2 \delta$. If there exists a $(k \times(n-k), M, \delta,[0, k-\delta])_q$-GRMC, then there exists an $(n, q^{(n-k)(k-\delta+1)}+M,2 \delta,k)_q$-CDC, which contains a lifted MRD code $(n, q^{(n-k)(k-\delta+1)}, 2 \delta, k)_q$-CDC as a subset.
	\end{lem}

	Below, we present the definitions of reduced row echelon form (RREF) and reduced row inverse echelon form (RRIEF).
		\begin{definition} \cite{4,19}\label{definition2}	A matrix is in RREF  (respectively RRIEF ) if:
		\begin{itemize}
			\item The leading coefficient of a row is always to the right (respectively left) of the leading coefficient of the previous row.
			\item All leading coefficients are ones.
			\item Every leading coefficient is the only nonzero entry in its column.
		\end{itemize}
	\end{definition} 
	
	A $k$-dimensional subspace $\mathcal{U}$ of $\mathbb{F}_q^n$ can be represented by a $k \times n$ generator matrix $U$ whose rows form a basis of $\mathcal{U}$. There exists a unique generator matrix $E(\mathcal{U})$ (respectively $\widehat{E}(\mathcal{U}))$ in RREF (respectively RRIEF).
	The\textit{ identifying} (respectively \textit{inverse identifying}) \textit{vector} $v(\mathcal{U})$ (respectively $\hat{v}(\mathcal{U}))$ of $\mathcal{U}$ is a binary vector of length $n$ and weight $k$, where the $k$ ones of positions are exactly the pivots of $E(\mathcal{U})$ (respectively $\hat{E}(\mathcal{U}))$.
	
		The \textit{echelon Ferrers form} (repectively\textit{ inverse echelon Ferrers form}) of an identifying vector $v$ (respectively an inverse identifying vector $\hat{v}$) of length $n$ and weight $k$, $EF(v)$ (respectively $\widehat{EF}(\hat{v})$) is the matrix in RREF (respectively RRIEF) with leading entries (of rows) in the columns indexed by the nonzero entries of $v$ (respectively $\hat{v}$) and $\bullet$ (called a dot) in all entries which do not have terminals zeros or ones. 
	
		\begin{example}\label{example1} 
			Consider the generator matrix of subspace  $\mathcal{U} \in \mathcal{G}_2(7,3)$ with
			RREF and RRIEF
			
			$$
			\begin{aligned}
				& E(\mathcal{U})=\left[\begin{array}{lllllll}
					\mathbf{1} & 1 & 0 & 0 & 0 & 0 & 1 \\
					0 & 0 & \mathbf{1} & 0 &1 & 0 & 1 \\
					0 & 0 & 0 &\mathbf{1}& 0 & 1 & 0
				\end{array}\right],
			\end{aligned}
		\begin{aligned}
			&\widehat{E}(\mathcal{U})=\left[\begin{array}{lllllll}
				1 &0 & 0 & 0 & 1 & 0 & \mathbf{1} \\
				1 & 0 & 0 & 0 & 1& \mathbf{1} & 0 \\
				1 & 0 & 0 & \mathbf{1} & 0 & 0 & 0
			\end{array}\right].
		\end{aligned}
			$$
		Its identifying vector is $v=(1011000)$, and inverse identifying vector $\hat{v}=(0001011)$. The echelon Ferrers form $E F(v)$ and the inverse echelon Ferrers form $\widehat{E F}(\hat{v})$ are as follows:
		$$
		\begin{aligned}
			& E F(v)=\left[\begin{array}{lllllll}
				\mathbf{1} & \bullet & 0 & 0 & \bullet & \bullet & \bullet \\
				0 & 0 & \mathbf{1} & 0 &\bullet & \bullet & \bullet \\
				0 & 0 & 0 &\mathbf{1}& \bullet & \bullet & \bullet
			\end{array}\right],
		\end{aligned}
		\begin{aligned}
			&\widehat{E F}(\hat{v})=\left[\begin{array}{lllllll}
				\bullet &\bullet & \bullet & 0 & \bullet & 0 & \mathbf{1} \\
				\bullet & \bullet & \bullet & 0 & \bullet & \mathbf{1} & 0 \\
				\bullet & \bullet & \bullet & \mathbf{1} & 0 & 0 & 0
			\end{array}\right].
		\end{aligned}
		$$
	\end{example}
	
	Based on the above, the authors provided the definition of a matrix in reduced row bilateral echelon form (RRBEF) in \cite{30}.

	\begin{definition} \cite{30}\label{definition3}
		An $m\times n$ matrix $\overline{E}(\mathcal{U})$ is said to be in RRBEF if its $m_1\times n_1$ submatrix in the upper left-hand corner is in RREF, its  $m_2\times n_2$ submatrix in the lower right-hand corner is in RRIEF, and the middle part is an $(m_1+m_2)\times (n-n_1-n_2)$ submatrix, where $m_1, m_2, n_1, n_2,m,n$ are positive integers, $m_1+m_2=m$ and $n_1+n_2 \leq n.$
	\end{definition}
	
	A \textit{bilateral identifying vector} is a binary vector of length $n$ and weight $k$, whose the subvector consisting of the first $n_1$ coordinates is a normal identifying vector and whose subvector consisting of the last $n_2$ coordinates is an inverse identifying vector, where $n_1, n_2, k,n$ are positive integers and $n_1+n_2 \leq n$, denoted by $\bar{v}=$ $\left(x_1, \dots, x_{n_1}, \bar{x}_{n_1+1}, \dots, \bar{x}_{n-n_2}, \hat{x}_{n-n_2+1}, \dots, \hat{x}_n\right)$. The ordered triple $(n_1, n-n_1-n_2, n_2)$
	is designated as the type of this bilateral identifying vector. 
	The \textit{bilateral echelon Ferrers form} of the bilateral identifying vector $\bar{v}$ is the $k \times n$ matrix in RRBEF such that the $k_1 \times n_1$ submatrix in the upper left corner is the echelon Ferrers form of $v_1=(x_1, \dots, x_{n_1})$, the $k_2 \times n_2$ submatrix in the lower right corner is the inverse echelon Ferrers form of $\hat{v}_2=\left(\hat{x}_{n-n_2+1}, \dots, \hat{x}_n\right)$, the $k \times (n-n_1-n_2)$ in the middle is filled with dots and the rest is 0, where $k_1,k_2$ are positive integers and $k_1+k_2=k.$

		\begin{example}\label{example3} 
		Consider the subspace $\mathcal{U} \in \mathcal{G}_2(14,6)$ and assume that a generator matrix in RRBEF of $\mathcal{U}$ is as follows:
		
		$$
		\begin{aligned}
			& \overline{E}(\mathcal{U})=\left[\begin{array}{llllllllllllll}
				\mathbf{1} & 0 & 1 & 0 & 0 & 0 & 1 & 0 & 0 & 0 & 0 & 0 & 0 & 0 \\
				0 & \mathbf{1} & 1 & 0 & 1 & 0 & 1 & 0 & 0 & 0 & 0 & 0 & 0 & 0 \\
				0 & 0 & 0 &\mathbf{1} & 1 & 1 & 0 & 0 & 0 & 0 & 0 & 0 & 0 & 0 \\
				0 & 0 & 0 & 0 & 0 & 0 & 0 & 1 & 1 & 0 & 0 & 0 & 1 & \mathbf{1} \\
				0 & 0 & 0 & 0 & 0 & 0 & 1 & 0 & 1 & 1 & 0 & \mathbf{1} & 0 & 0 \\
				0 & 0 & 0 & 0 & 0 & 0 & 1& 0 & 0 & 0 & \mathbf{1} & 0 & 0 & 0
			\end{array}\right].
		\end{aligned}
		$$
		Its bilateral identifying vector is $\bar{v}=(110100 \bar{0} \bar{0} \hat{0} \hat{0} \hat{1} \hat{1} \hat{0} \hat{1})$,  and the bilateral echelon Ferrers form $\overline{E F}(\bar{v})$ is as follows:
		$$
		\begin{aligned}
			& \overline{E F}(\bar{v})=\left[\begin{array}{llllllllllllll}
				\mathbf{1} & 0 & \bullet & 0 & \bullet & \bullet & \bullet & \bullet & 0 & 0 & 0 & 0 & 0 & 0 \\
				0 & \mathbf{1} & \bullet & 0 & \bullet & \bullet & \bullet & \bullet & 0 & 0 & 0 & 0 & 0 & 0 \\
				0 & 0 & 0 &\mathbf{1} & \bullet & \bullet & \bullet & \bullet & 0 & 0 & 0 & 0 & 0 & 0 \\
				0 & 0 & 0 & 0 & 0 & 0 & \bullet & \bullet & \bullet & \bullet & 0 & 0 & \bullet & \mathbf{1} \\
				0 & 0 & 0 & 0 & 0 & 0 & \bullet & \bullet & \bullet & \bullet & 0 & \mathbf{1} & 0 & 0 \\
				0 & 0 & 0 & 0 & 0 & 0 & \bullet& \bullet & \bullet & \bullet & \mathbf{1} & 0 & 0 & 0
			\end{array}\right].
		\end{aligned}
		$$

	\end{example}

	An $m \times n$ \textit{Ferrers diagram} $\mathcal{F}$ is an $m × n$ array of dots and empty entries such that 
	 all dots are shifted to the right,
	 the number of dots in each row is less than or equal to the number of dots in the previous row,
	 and the ﬁrst row has $n$ dots and the rightmost column has $m$ dots.

	For a given $m\times n$ Ferrers diagram $\mathcal{F}$, we denote by $|\mathcal{F}|$ the number of dots in $\mathcal{F}$, and we call it a full Ferrers diagram if the number of dots in $\mathcal{F}$ is $mn$.
	In general, the number of dots in the $i$-th column of $\mathcal{F}$ is denoted by $\gamma_i$ and the number of dots in the $i$-th row of $\mathcal{F}$ is denoted by $\rho_i$. There is a unique $m \times n$ Ferrers diagram $\mathcal{F}$  such that the $i$-th column of $\mathcal{F}$ has cardinality $\gamma_i$ for any $0 \leq i \leq n-1$. On this condition, we write $\mathcal{F}=\left[\gamma_0, \gamma_1, \cdots, \gamma_{n-1}\right]$, its inverse Ferrers diagram $\hat{\mathcal{F}}=\left[\gamma_{n-1}, \gamma_{n-2}, \cdots, \gamma_0\right]$, and its transposed Ferrers diagram $\mathcal{F}^t=$ $\left[\rho_{m-1}, \rho_{m-2}, \cdots, \rho_0\right]$.
	
	\begin{example}\label{example2} Let $\mathcal{F}=[2,3,4,5]$, then\\
		
		\centering 
		$
		\mathcal{F}=\begin{matrix}
			&\bullet& \bullet& \bullet&\bullet\\
			&\bullet&\bullet&\bullet&\bullet\\
			& &\bullet&\bullet&\bullet\\
			& & &\bullet&\bullet\\
			& & & &\bullet
		\end{matrix},$ \enspace    
		$
		\mathcal{\hat{F}}=\begin{matrix}
			&\bullet&\bullet&\bullet&\bullet\\
			&\bullet&\bullet&\bullet&\bullet\\
			&\bullet&\bullet&\bullet\\
			&\bullet&\bullet\\
			&\bullet
		\end{matrix},$ \enspace    
		$
		\mathcal{F}^t=\begin{matrix}
			&\bullet&\bullet&\bullet&\bullet&\bullet\\
			&&\bullet&\bullet&\bullet&\bullet\\
			&&&\bullet&\bullet&\bullet\\
			&&&&\bullet&\bullet
		\end{matrix} .
		$
	\end{example}

		Let $\mathcal{F}$ be an $m \times n$ Ferrers diagram, an $[\mathcal{F}, k, \delta]_q$ \textit{Ferrers diagram rank-metric code (FDRMC)} is abbreviated as $[\mathcal{F}, k, \delta]_q$ code, is an $[m \times n, k, \delta]_q$ RMC in which for each $m \times n$ matrix, all entries not in $\mathcal{F}$ are zero. Note that if there exists an $[\mathcal{F}, k, \delta]_q$ code, then so does an $[\hat{\mathcal{F}}, k, \delta]_q$ code and an $\left[\mathcal{F}^t, k, \delta\right]_q$ code. Etzion and Silberstein gave a Singleton-like upper bound on FDRMCs in \cite{4}.
		
		\begin{lem}{\cite{4}}\label{lem6} Let $\delta$ be a positive integer. Let $v_i$, $0 \leq i \leq \delta-1$, be the number of dots in a Ferrers diagram $\mathcal{F}$ which are not contained in the first $i$ rows and the rightmost $\delta-1-i$ columns. Then for any $[\mathcal{F}, k, \delta]_q$ code, $k \leq \min _{i \in\{0,1, \cdots, \delta-1\}} v_i$.
		\end{lem}
		An FDRMC attaining the bound in Lemma 6 is referred to as optimal. Constructions for optimal FDRMCs can be discovered in \cite{1,5,8,20,21,31}. We are currently only mentioning the construction that relies on the subcodes of MRD codes for the future use.
		
		\begin{lem}{\cite{5}}\label{lem7} Suppose that $\mathcal{F}$ is an $m \times n$ $(m \geq n)$ Ferrers diagram and each of the rightmost $\delta-1$ columns of $\mathcal{F}$ has at least $n$ dots. Then there exists an optimal $[\mathcal{F}, \sum_{i=0}^{n-\delta} \gamma_i, \delta]_q$ code for any prime power $q$.
		\end{lem}
Similar to the definitions of Ferrers diagram and FDRMCs, bilateral Ferrers diagram and bilateral Ferrers diagram rank-metric codes were given in \cite{30}.

	A $(k_1+ k_2)\times n$ \textit{bilateral Ferrers diagram} $\mathcal{F}$ is a $(k_1 + k_2)\times n$ array of dots and empty cells such that the $k_1\times n_1$  subarray in the upper left-hand corner is a $k_1\times n_1$ Ferrers diagram, the $k_2\times n_2$  subarray in the lower right-hand corner subarray 
	is a $k_2\times n_2$ inverse Ferrers diagram. the $(k_1+ k_2)\times (n-n_1-n_2)$ subarray in the middle is a full Ferrers diagram, and that the remaining cells are empty, where $n_1+n_2 \leq n$. Let $\mathcal{F}$ be an $m \times n$  bilateral Ferrers diagram, an $[\mathcal{F}, k, \delta]_q$ \textit{bilateral Ferrers diagram rank-metric code (BFRMC)} is an $[m \times n, k, \delta]_q$ RMC in which for each $m\times n$ matrix, all entries not in $\mathcal{F}$ are zero.

	In \cite{4}, the authors gave the multilevel construction, which is a generalization of Lemma 2. Moreover, the multilevel construction was generalized to the inverse multilevel construction in \cite{19}. The following lemmas are crucial for the multilevel and inverse multilevel construction.
	
	\begin{lem}{\cite{4}}\label{lem8} Let $\mathcal{U}$, $\mathcal{V}$ $\in\mathcal{G}_q(n, k)$, $\mathcal{U}={rs}(U)$ and  $\mathcal{V}={rs}(V)$, where $U$, $V$ $\in \mathbb{F}_q^{k \times n}$ are in RREF, then
		$$
		d_S(\mathcal{U}, \mathcal{V}) \geq d_H(v(\mathcal{U}), v(\mathcal{V})).
		$$
		Furthermore, if $v(\mathcal{U})=v(\mathcal{V})$, then
		$$
		d_S(\mathcal{U}, \mathcal{V})=2 d_R\left(C_U, C_V\right),
		$$
		where $C_U$ and $C_V$ denote the submatrices of $U$ and $V$ without the columns of their pivots, respectively.
	\end{lem}
	
	\begin{lem}{\cite{19}}\label{lem9} Let $\mathcal{U}$, $\mathcal{V}$ $\in\mathcal{G}_q(n, k)$, $\mathcal{U}={rs}(U)$ and  $\mathcal{V}={rs}(V)$, where $U$, $V$ $\in \mathbb{F}_q^{k \times n}$ are in RRIEF, then
		$$
		d_S(\mathcal{U}, \mathcal{V}) \geq d_H(\hat{v}(\mathcal{U}), \hat{v}(\mathcal{V})).
		$$
		Furthermore, if $\hat{v}(\mathcal{U})=\hat{v}(\mathcal{V})$, then
		$$
		d_S(\mathcal{U}, \mathcal{V})=2 d_R\left(C_U, C_V\right),
		$$
		where $C_U$ and $C_V$ denote the submatrices of $U$ and $V$ without the columns of their pivots, respectively.
	\end{lem}
	
	In \cite{19}, the double multilevel construction was given by combining the multilevel and inverse multilevel construction.
	
	\begin{lem}{\cite{19}}\label{lem10} (\textbf{Double multilevel construction}) Let $n \geq 2 k \geq 2 \delta$. Suppose that $\mathcal{C}_1$ is
		an $\left(n, M_1, 2 \delta, k\right)_q$-CDC constructed by the multilevel construction with the identifying vector set $\mathcal{A}$, and $\mathcal{C}_2$ is an $\left(n, M_2, 2 \delta, k\right)_q$-CDC constructed by the inverse multilevel construction with the inverse identifying vector set $\hat{\mathcal{A}}$. If each inverse identifying vector $\hat{v} \in \hat{\mathcal{A}}$ satisfies $d_{\mathrm{H}}(\hat{v}, v) \geq 2\left(s_{\hat{v}}+\delta\right)$ for any identifying vector $v \in \mathcal{A}$, where $s_{\hat{v}}$ is an integer related to $\hat{v}$, then there exists an $\left(n, M_1+M_2, 2 \delta, k\right)$-CDC.
	\end{lem}
	
	Likewise Lemma 8 and Lemma 9, the authors provided the following lemma in \cite{30}, which is beneficial for the bilateral multilevel construction.
	
	\begin{lem} \cite{30} \label{lem11}Let $\mathcal{U}$, $\mathcal{V}$ $\in \mathcal{G}_q(n, k)$, $\mathcal{U}={rs}(U)$ and $\mathcal{V}={rs}(V)$, where $U$, $V$ $\in \mathbb{F}_q^{k \times n}$ are in RRBEF. If two bilateral identifying vectors $\bar{v}(\mathcal{U})$, $\bar{v}(\mathcal{V})$ of $U$ and $V$ have the same type $(n_1, n-n_1-n_2, n_2)$ with $n_1+n_2 \leq n$, then
		$$
		d_S(\mathcal{U}, \mathcal{V}) \geq d_H( \bar{v}(\mathcal{U}), \bar{v}(\mathcal{V})).
		$$
		Furthermore, if $\bar{v}(\mathcal{V})=\bar{v}(\mathcal{U})$, then
		$$
		d_S(\mathcal{U}, \mathcal{V})=2 d_R\left(C_U, C_V\right),
		$$
		where $C_U$ and $C_V$ denote the submatrices of $U$ and $V$ without the columns of their pivots, respectively.
	\end{lem}
	
Based on Lemma 11, the authors provided the bilateral multilevel construction in \cite{30} as follows.
	
	\begin{lem}{\cite{30}}\label{lem12} (\textbf{Bilateral multilevel construction}) Let $\mathcal{B}$ be a set of bilateral identifying vectors, which is a binary constant weight code of length $n$, weight $k$ and minimum Hamming distance $2 \delta$. Suppose that all bilateral identifying vectors have the same type. For each codeword $\bar{v} \in \mathcal{B}$, its bilateral echelon Ferrers form is $\overline{E F}(\bar{v})$, all dots in $\overline{E F}(\bar{v})$ produce a bilateral Ferrers diagram $\mathcal{F}_{\bar{v}}$. If there exists an $\left(\mathcal{F}_{\bar{v}}, M_{\bar{v}}, \delta\right)_q$ BFDRMC $\mathcal{D}_{\bar{v}}$ for each $\bar{v} \in \mathcal{B}$, and $\mathbb{D}_{\bar{v}}$ represents the lifted BFDRM code of $\mathcal{D}_{\bar{v}}$. Then $\mathcal{C}_3=\bigcup_{\bar{v} \in \mathcal{B}} \mathbb{D}_{\bar{v}}$ is an $(n, 2 \delta, k)_q$-CDC.
	\end{lem}

	The following lemma is useful for calculating the size of some CDCs.
	
	\begin{lem}{\cite{18}}\label{lem13} Let $n \geq 2 k+\delta,~k \geq 2 \delta$ and
		$$
		M_1=q^{(n-k)(k-\delta+1)} \frac{1-q^{-\lfloor\frac{k}{\delta}\rfloor\delta^2}}{1-q^{-\delta^2}}+q^{(n-k-\delta)(k-\delta+1)} .
		$$
		Then there exists an $\left(n, M_1, 2 \delta, k\right)_q$-CDC constructed via the multilevel construction with the identifying vector set $\mathcal{A}$, and the CDC contains a lifted MRD code $\left(n, q^{(n-k)(k-\delta+1)}, 2 \delta, k\right)_q$-CDC as a subset, where
		$$
		\mathcal{A}=\{(\underbrace{1 \cdots 1}_k 0 \cdots 0)\} \cup\{(\underbrace{1 \cdots 1}_{k-i \delta} \underbrace{0 \cdots 0}_\delta \underbrace{1 \cdots 1}_{i \delta} 0 \cdots 0) \mid 1 \leq i \leq\lfloor\frac{k}{\delta}\rfloor\} .
		$$
	\end{lem}
	
	The subcode construction is important for improving the lower bounds of size for CDCs under our construction as we provide below.
	\begin{lem} \cite{cos} \label{lem14}(\textbf{Subcode Construction}) Let $\mathcal{C}$ be a linear $[m \times n, k, \delta]_q$ RMC, $\mathcal{M}$ a linear $\left[m \times n, k^{\prime}, \delta^{\prime}\right]_q$ subcode of $\mathcal{C}$ where $\delta^{\prime} \geq \delta$. Then $s=q^{k-k^{\prime}}$ RMCs satisfying the following conditions can be constructed:\\
		(1) $\mathcal{M}_j\left(m \times n, k^{\prime}, \delta^{\prime}\right)_q$ is RMC for all $1 \leq j \leq s$;
		\\(2) For $M \in \mathcal{M}_j$ and $M^{\prime} \in \mathcal{M}_{j^{\prime}}$, satisfy $M \neq M^{\prime}$ and ${rank}\left(M-M^{\prime}\right) \geq \delta$ for all $1 \leq j<j^{\prime} \leq s.$ 
	\end{lem}
	
	For convenience, we  adopt the following notations in this paper.
	
	\begin{itemize}
		\item $rs(A)$ denote the subspace in $\mathbb{F}_q^{n}$ spanned by the rows of
		a matrix $ A \in \mathbb{F}_q^{k \times n}. $
		\item $E(\mathcal{U})$ denote a generator matrix in RREF of subspace $\mathcal{U}$.
		\item $\widehat{E}(\mathcal{U})$ denote a generator matrix in RRIEF of subspace $\mathcal{U}$.
		\item $\overline{E}(\mathcal{U})$ denote a generator matrix in RRBEF of subspace $\mathcal{U}$.
		
		\item $v(\mathcal{U})$ denote the identifying vector of a generator matrix in RREF of subspace $\mathcal{U}$.
		\item $\hat{v}(\mathcal{U})$ denote the inverse identifying vector   of a generator matrix in RRIEF of subspace $\mathcal{U}$.
		\item $\bar{v}(\mathcal{U})$ denote the bilateral identifying vector of a generator matrix in RRBEF of subspace $\mathcal{U}$.
		\item $EF(v)$ denote the echelon Ferrers form of  identifying vector $v$.
		\item $\widehat {EF}(\hat{v} )$ denote the inverse echelon Ferrers form of  inverse identifying vector $\hat{v}$.
		\item $\overline{EF}(\bar{v})$ denote the bilateral echelon Ferrers form of  bilateral identifying vector $\bar{v}$.
		
	\end{itemize}

	\section{Inverse bilateral multilevel construction}
	In this section, we propose the inverse bilateral multilevel construction by introducing inverse bilateral identifying vectors and inverse bilateral Ferrers diagram rank-metric codes.	
	
We denote by $\hat{I}_k$ the $k\times k$ matrix defined as follow:	
$$	
\left[\begin{array}{ccccc}
	0 & 0 & \cdots & 0 & 1\\
	0 & 0 & \cdots & 1 & 0\\
	\vdots & \vdots &\ddots&\vdots& \vdots\\
	1& 0& \cdots & 0 & 0
\end{array}\right].
$$

	An $m\times n$ matrix is said to be in \textit{reduced row inverse bilateral echelon form} if its $m_1\times n_1$ submatrix in the lower left-hand corner is in RREF, its  $m_2\times n_2$ submatrix in the upper right-hand corner is in RRIEF, and the middle part is an $(m_1+m_2)\times (n-n_1-n_2)$ submatrix, where  $m_1, m_2, n_1, n_2,m,n$ are positive integers, $m_1+m_2=m$ and $n_1+n_2 \leq n.$

To differentiate between an \textit{inverse bilateral identifying vector} and a bilateral identifying vector, we adopt the notation $\hat{\bar{v}} =$$ (\tilde{x}_1,\dots,\tilde{x}_{n_1}, 
	\bar{\tilde{x}}_{n_1+1},\dots,\bar{\tilde{x}}_{n-n_2},\hat{\tilde{x}}_{n-n_2+1},\dots,\hat{\tilde{x}}_n)$ to represent an inverse bilateral identifying vector of lengh $n$ and weight $k$, where $n_1,n_2,k,n$ are positive integers and  $n_1+n_2\leq n$. The ordered triple $(n_1,n-n_1-n_2,n_2)$
	is called the type of this inverse bilateral identifying vector.
	The \textit{inverse bilateral echelon Ferrers form} of the inverse bilateral identifying vector $\hat{\bar{v}}$, $\widehat{\overline{EF}}(\hat{\bar{v}})$, is the $k \times n$ matrix in redeced row inverse bilateral echelon form such that the $\delta_1 \times n_1$ submatrix in the lower left corner is the echelon Ferrers form of $\tilde{v}_1=(\tilde{x}_1,\dots,\tilde{x}_{n_1}) 
	$, the $\delta_2 \times n_2$ submatrix in the upper right corner is the inverse echelon Ferrers form of $\hat{\tilde{v}}_2=$$(\hat{\tilde{x}}_{n-n_2+1},\dots,\hat{\tilde{x}}_n)$, the $k \times (n-n_1-n_2)$ in the middle is filled with dots and the rest is 0, where $\delta_1,\delta_2$ are positive integers and $\delta_1+\delta_2=k.$	
	
	Assume $\hat{\bar{v}}$ is an inverse bilateral identifying vector as follows:
	$$
	\begin{gathered}
		\hat{\bar{v}}=(\overbrace{\underbrace{\tilde{0}\cdot\cdot\cdot\tilde{0}}_{m_1}\underbrace{\tilde{1}\cdots\tilde{1}}_{\delta_1}\underbrace{\tilde{0}\cdots\tilde{0}}_{m_2}}^{n_1} \overbrace{\bar{\tilde{0}}\cdots\bar{\tilde{0}}}^{n-n_1-n_2} \overbrace{\underbrace{\hat{\tilde{0}}\cdots \hat{\tilde{0}}}_{m_3}\cdots\underbrace{\hat{\tilde{1}}\cdots\hat{\tilde{1}}}_{\delta_2} \underbrace{\hat{\tilde{0}}\cdots\hat{\tilde{0}}}_{m_4}}^{n_2}),
	\end{gathered}
	$$
	where $n_1+n_2\leq n$, $m_1$, $m_2$, $m_3$, $m_4$ are nonnegetive integers and $\delta_1+\delta_2=k.$ The {inverse bilateral echelon Ferrers form} of $\hat{\bar{v}}$ is as follows:
	$$
	\widehat{\overline{EF}}(\hat{\bar{v}})=\left[\begin{array}{ccccccc}
		0_1 & 0_2 & 0_3 & \mathcal{F}_3&\mathcal{F}_4 &\hat{I}_{\delta_2}& 0_4\\
		0_5& {I}_{\delta_1}&  \mathcal{F}_1&  \mathcal{F}_2& 0_6& 0_7& 0_8
		
	\end{array}\right],
	$$
	where $\mathcal{F}_1$ is a $\delta_1\times m_2$ full Ferrers diagram, $\mathcal{F}_2$ is a $\delta_1\times (n-n_1-n_2)$ full Ferrers diagram, $\mathcal{F}_3$ is a $\delta_2\times (n-n_1-n_2)$ full Ferrers diagram diagram and  $\mathcal{F}_4$
	is a $\delta_2\times m_3$ full Ferrers diagram, $0_1,\dots,0_8$ are zero matrices.

	\begin{example}\label{example4}
		Consider the subspace $\mathcal{U} \in \mathcal{G}_2(15,6)$ and assume that a generator matrix in reduced row inverse bilateral echelon form of $\mathcal{U}$ is as follows:
		$$
		\widehat{\overline{E}}(\mathcal{U})=\left[\begin{array}{lllllllllllllll}
			0 & 0 & 0 & 0 & 0 & 0 & 0 & 1 & 0 &1 & 0 & 0 & 0 & 0 &\mathbf{1} \\
			0 & 0 & 0 & 0 & 0 & 0 & 0 & 1 & 1 &  1 & 1 & 0 & \mathbf{1} & 0 & 0 \\
			0 & 0 & 0 & 0 & 0 & 0 & 0 & 0 & 0 & 0 & 1 &\mathbf{1}& 0 & 0 & 0\\
			\mathbf{1}& 0 & 0 & 0 & 1 & 1 & 0 & 0 & 1 &  0 & 0 & 0 & 0 & 0 & 0\\
			0 & 0 & \mathbf{1} & 0 & 1 & 0 & 1 & 1 & 0 &  0 & 0 & 0 & 0 & 0 & 0\\
			0 & 0 & 0 & \mathbf{1} & 0 & 1 & 1 & 1 & 1 & 0 & 0 & 0 & 0 & 0 & 0
		\end{array}\right].
		$$
		Let the inverse bilateral identifying vector of   $\widehat{\overline{E}}(\mathcal{U})$ be $\hat{\bar{v}}$, then $\hat{\bar{v}}=(\tilde{1}\tilde{0}\tilde{1}\tilde{1}\tilde{0}\tilde{0}\tilde{0}\bar{\tilde{0}}  \bar{\tilde{0}} \hat{\tilde{0}} \hat{\tilde{0}} \hat{\tilde{1}} \hat{\tilde{1}} \hat{\tilde{0}} \hat{\tilde{1}})$ and the inverse bilateral echelon Ferrers form of $\hat{\bar{v}}$ is
		
		$$
		\widehat{\overline{EF}}(\hat{\bar{v}})=\left[\begin{array}{ccccccccccccccc}
			0 & 0 & 0 & 0 & 0 & 0 & 0 & \bullet & \bullet & \bullet & \bullet & 0 & 0 & \bullet & \mathbf{1} \\
			0 & 0 & 0 & 0 & 0 & 0 & 0 & \bullet & \bullet & \bullet &\bullet & 0 &\mathbf{1} & 0 & 0 \\
			0 & 0 & 0 & 0 & 0 & 0 & 0 & \bullet & \bullet & \bullet & \bullet & \mathbf{1} & 0 & 0 & 0\\
			\mathbf{1}& \bullet & 0 & 0 & \bullet & \bullet & \bullet & \bullet & \bullet & 0 & 0 & 0 & 0 & 0 & 0 \\
			0 & 0 &\mathbf{1} & 0 & \bullet & \bullet & \bullet & \bullet & \bullet & 0 & 0 & 0 & 0 & 0 & 0 \\
			0 & 0 & 0 & \mathbf{1}& \bullet & \bullet & \bullet & \bullet & \bullet & 0 & 0 & 0 & 0 & 0 & 0 \\
		\end{array}\right],
		$$
		and the inverse bilateral Ferrers diagram is 
		
		$$
		\mathcal{F}=\begin{matrix}
			& & & & &\bullet&\bullet&\bullet& \bullet&\bullet\\
			& & & & & \bullet&\bullet&\bullet& \bullet\\
			& & & & &\bullet&\bullet&\bullet& \bullet\\
			&\bullet& \bullet& \bullet&\bullet&\bullet&\bullet\\
			& &\bullet&\bullet&\bullet&\bullet&\bullet\\
			& &\bullet&\bullet&\bullet&\bullet&\bullet
		\end{matrix}.$$

	\end{example}

	We can give the following corollary by Lemma 11, which is crucial for controlling the subspace diatance within the inverse bilateral multilevel construction.
	
	\begin{coro} \label{coro1}Let $\mathcal{U}$, $\mathcal{V}$ $\in \mathcal{G}_q(n, k)$, $\mathcal{U}={rs}(U)$ and $\mathcal{V}={rs}(V)$, where $U$, $V$ $\in \mathbb{F}_q^{k \times n}$ are in reduced row inverse bilateral echelon forms. If two inverse bilateral identifying vectors $\hat{\bar{v}}(\mathcal{U})$, $\hat{\bar{v}}(\mathcal{V})$ of $U$ and $V$ have the same type $\left(n_1, n-n_1-n_2, n_2\right)$ with $n_1+n_2 \leq n$, then
		$$
		d_S(\mathcal{U}, \mathcal{V}) \geq d_H(\hat{\bar{v}}(\mathcal{U}),\hat{\bar{v}}(\mathcal{V})).
		$$
		Furthermore, if $\hat{\bar{v}}(\mathcal{U})=\hat{\bar{v}}(\mathcal{V})$, then 
		$$
		d_S(\mathcal{U}, \mathcal{V})=2 d_R\left(C_U, C_V\right),
		$$
		where $C_U$ and $C_V$ denote the submatrices of $U$ and $V$ without the columns of their pivots, respectively.
		
	\end{coro}

	\begin{theo}\label{theo1} (\textbf{Inverse bilateral multilevel construction}) Let $\hat{\mathcal{B}}$ be a set of inverse bilateral identifying vectors, which is a binary constant weight code of length $n$, weight $k$ and minimum Hamming distance $2 \delta$. Suppose that all inverse bilateral identifying vectors have the same type. For each codeword $\hat{\bar{v}} \in \hat{\mathcal{B}}$, its inverse bilateral echelon Ferrers form is $\widehat{\overline{E F}}(\hat{\bar{v}})$, all dots in $\widehat{\overline{E F}}(\hat{\bar{v}})$ produce an inverse bilateral Ferrers diagram $\mathcal{F}_{\hat{\bar{v}}}$. If there exists an $\left(\mathcal{F}_{\hat{\bar{v}}}, M_{\hat{\bar{v}}}, \delta\right)_q$ inverse bilateral Ferrers diagram rank-metric code  $\mathcal{D}_{\hat{\bar{v}}}$ for each $\hat{\bar{v}} \in \hat{\mathcal{B}}$, the row spaces of the matrices in $\bigcup_{\hat{\bar{v}} \in \hat{\mathcal{B}}} \mathcal{D}_{\hat{\bar{v}}}$ form an $(n, 2 \delta, k)_q$-CDC.
	\end{theo}
	\begin{proof}
		Clearly, it is sufficient to prove the subspace distance. It is easy to see from Corollary 1 that for any $\mathcal{U},$ $\mathcal{V}$ in the CDC mentioned above and $\mathcal{U}\neq\mathcal{V}$, we have $d_S(\mathcal{U}, \mathcal{V}) \geq 2\delta$.
	\end{proof}
	
	\section{Multilevel inserting construction}
	In this section, based on subcode construction and inverse bilateral multilevel construction, we expand the cardinality of some CDCs described in \cite{30}. Furthermore, by introducing a new set of bilateral identifying vectors, we improve the lower bounds of $\bar{A}_q(17,6,6)$, $\bar{A}_q(18,6,7)$ and $\bar{A}_q(19,6,7)$. 
	
	Below, based on the idea of Lemma 14, we give the following lemmas.
	\begin{lem}\label{lem15} Let $n_1$, $n_2$, $n_3$, $\delta_1$, $\delta_2$, $b_1$, $b_2$, $k_1$, $k_2$ be positive integers satisfying $n_1+n_2+n_3=n$, $\delta_1+\delta_2=k$, $b_1+b_2 \geq \delta$ and $1<b_i<\delta$ for $i=1,2$. Presume that $\mathcal{D}_1$ is a linear $\left[ \delta_1\times\left(n_1-\delta_1\right), k_1,b_1\right]_q$ RMC, $\mathcal{M}_1$ is a linear $\left[ \delta_1 \times\left(n_1-\delta_1\right), {k}_1^{\prime},\delta\right]_q$ subcode of $\mathcal{D}_1$. Assume that  $\mathcal{D}_2$ is a linear  $\left[ \delta_2 \times n_3, k_2, b_2\right] _q$ RMC, $\mathcal{M}_2$ is a linear $\left[ \delta_2 \times n_3, {k}_2^{\prime},\delta\right] _q$ subcode of $\mathcal{D}_2$. Then $s=\min\{{q}^{k_1-{k}_1^{\prime}},{q}^{k_2-{k}_2^{\prime}}\}$ RMCs satisfying the following conditions can be constructed:
	\begin{itemize}
		\item $\mathcal{M}_1^r \left(\delta_1 \times\left(n_1-\delta_1\right),\left|\mathcal{M}_1^r\right|, \delta\right)_q$ and 
		$\mathcal{M}_2^r \left(\delta_2 \times n_3,\left|\mathcal{M}_2^r\right|, \delta\right)_q$ are RMCs for all  $1 \leq r\leq s$;
		\item For $M \in \mathcal{M}_i^r$ and $M^{\prime} \in \mathcal{M}_i^{r^{\prime}}$ satisfy $M \neq M^{\prime}$ and ${rank}\left(M-M^{\prime}\right) \geq b_i$ for all $1\leq i\leq 2$, $1 \leq r<r^{\prime} \leq s $. 
	\end{itemize}
	Let $\mathcal{M}_3$ is a $\left(\delta_1 \times n_3,\left|\mathcal{M}_3\right|, \delta\right)_q$ RMC. 
		  Then $\mathcal{C}_4=\bigcup_{r=1}^s \mathcal{C}_r$ is an $(n, *, 2 \delta, k)_q$-CDC, where 
		$$
		\begin{aligned}
		\mathcal{C}_r =	\left\{{ rs }\left[\begin{array}{ccccc}
				0_{\delta_2 \times \delta_1} & 0_{\delta_2 \times\left(n_1 - \delta_1\right)} & M_2 & \hat{I}_{\delta_2} & 0_{\delta_2 \times\left(n_2-\delta_2\right)}\\
				I_{\delta_1} & M_1 & M_3 & 0_{\delta_1 \times \delta_2} & 0_{\delta_1 \times\left(n_2-\delta_2\right)} 
			\end{array}\right]:
			M_1 \in \mathcal{M}_1^r, M_2 \in \mathcal{M}_2^r, M_3 \in \mathcal{M}_3\right\}.
		\end{aligned}
		$$
	\end{lem}
	
	\begin{proof}Clearly, $\mathcal{U}$ is a $k$-dimensional subspace of $\mathbb{F}_q^n$ for any $\mathcal{U} \in \mathcal{C}_4$. For the distance analysis, let $\mathcal{U}_1 \in \mathcal{C}_r$, $\mathcal{U}_2 \in \mathcal{C}_{r^{\prime}}$ be $k$-dimensional subspace for $1 \leq r \leq r^{\prime} \leq s$ and $\mathcal{U}_i$ be spanned by the rows of matrix $G_i$ for $i=1,2$,
		$$
		\begin{aligned}
			G_1 & =\left[\begin{array}{ccccc}
				0_{\delta_2 \times \delta_1} & 0_{\delta_2 \times\left(n_1 - \delta_1\right)} & M_2 & \hat{I}_{\delta_2} & 0_{\delta_2 \times\left(n_2-\delta_2\right)}\\
				I_{\delta_1} & M_1 & M_3 & 0_{\delta_1 \times \delta_2} & 0_{\delta_1 \times\left(n_2-\delta_2\right)}
			\end{array}\right], \\
			G_2 & =\left[\begin{array}{ccccc}
				0_{\delta_2 \times \delta_1} & 0_{\delta_2 \times\left(n_1 -\delta_1\right)} & M_2^{\prime} & \hat{I}_{\delta_2} & 0_{\delta_2 \times\left(n_2-\delta_2\right)}\\
				I_{\delta_1} & M_1^{\prime} & M_3^{\prime} & 0_{\delta_1 \times \delta_2} & 0_{\delta_1 \times\left(n_2-\delta_2\right)} 
			\end{array}\right],
		\end{aligned}
		$$
		where $M_1 \in \mathcal{M}_1^r$, $M_2 \in \mathcal{M}_2^r$, $M^{\prime}_1 \in \mathcal{M}_1^{r^{\prime}}$, $M^{\prime}_2 \in \mathcal{M}_2^{r^{\prime}}$ and $M_3$, $M^{\prime}_3 \in \mathcal{M}_3$. Then we have
		
		\begin{align*}
			&d_S\left(\mathcal{U}_1, \mathcal{U}_2\right)\\
			&=2\mathrm{rank}\left[\begin{array}{ccccc}
				0_{\delta_2 \times \delta_1} & 0_{\delta_2 \times\left(n_1 -\delta_1\right)} & M_2 & \hat{I}_{\delta_2} & 0_{\delta_2 \times\left(n_2-\delta_2\right)} \\
			I_{\delta_1} & M_1 & M_3 & 0_{\delta_1 \times \delta_2} & 0_{\delta_1 \times\left(n_2-\delta_2\right)} \\
				0_{\delta_2 \times \delta_1} & 0_{\delta_2 \times\left(n_1 - \delta_1\right)} & M_2^{\prime} & \hat{I}_{\delta_2} & 0_{\delta_2 \times\left(n_2-\delta_2\right)}\\
				I_{\delta_1} & M_1^{\prime} & M_3^{\prime} & 0_{\delta_1 \times \delta_2} & 0_{\delta_1 \times\left(n_2-\delta_2\right)} 
			\end{array}\right]-2 k\\
			&=2\mathrm{rank}\left[\begin{array}{ccccc}
				0_{\delta_2 \times \delta_1} & 0_{\delta_2 \times\left(n_1 -\delta_1\right)} & M_2 & \hat{I}_{\delta_2} & 0_{\delta_2 \times\left(n_2-\delta_2\right)} \\
				I_{\delta_1} & M_1 & M_3 & 0_{\delta_1 \times \delta_2} & 0_{\delta_1 \times\left(n_2-\delta_2\right)} \\
				0_{\delta_2 \times \delta_1} & 0_{\delta_2 \times\left(n_1 -\delta_1\right)} & M_2^{\prime}-M_2 & 0_{\delta_2 \times \delta_2} & 0_{\delta_2 \times\left(n_2-\delta_2\right)}\\
				0_{\delta_1 \times \delta_1} & M_1^{\prime}-M_1 & M_3^{\prime}-M_3 & 0_{\delta_1 \times \delta_2} & 0_{\delta_1 \times\left(n_2-\delta_2\right)} 
			\end{array}\right]-2 k\\
			&=2\mathrm{rank}\left[\begin{array}{ccccc}
				0_{\delta_2 \times \delta_1} & 0_{\delta_2 \times\left(n_1-\delta_1\right)} & 0_{\delta_2 \times n_3} & \hat{I}_{\delta_2} & 0_{\delta_2 \times\left(n_2-\delta_2\right)} \\
				I_{\delta_1} & 0_{\delta_1 \times\left(n_1-\delta_1\right)} & 0_{\delta_1 \times n_3} & 0_{\delta_1 \times \delta_2} & 0_{\delta_1 \times\left(n_2-\delta_2\right)} \\
				0_{\delta_2 \times \delta_1} & 0_{\delta_2 \times\left(n_1-\delta_1\right)} & M^{\prime}_2-M_2 & 0_{\delta_2 \times \delta_2} & 0_{\delta_2 \times\left(n_2-\delta_2\right)}\\
				0_{\delta_1 \times \delta_1} & M^{\prime}_1-M_1 & M^{\prime}_3-M_3 & 0_{\delta_1 \times \delta_2} & 0_{\delta_1 \times\left(n_2-\delta_2\right)} 
			\end{array}\right]-2 k \\
			&=2 \mathrm{rank}\left[\begin{array}{cc}
				0_{\delta_2 \times\left(n_1-\delta_1\right)} & M^{\prime}_2-M_2\\
				M^{\prime}_1-M_1 & M^{\prime}_3-M_3 
			\end{array}\right] .
		\end{align*}
		
		If $M^{\prime}_3\neq M_3$, then  $d_S\left(\mathcal{U}_1, \mathcal{U}_2\right) \geq 2 \mathrm{rank}\left(M^{\prime}_3-M_3\right) \geq 2 \delta$. If $M^{\prime}_3=M_3$ and $r=r^{\prime}$, then we have $M^{\prime}_1 \neq M_1$ or $M^{\prime}_2 \neq M_2$. Clearly, $d_S\left(\mathcal{U}_1, \mathcal{U}_2\right) \geq 2 \mathrm{rank}\left(M^{\prime}_1-M_1\right) \geq 2 \delta$ or $d_S\left(\mathcal{U}_1, \mathcal{U}_2\right) \geq 2\mathrm{rank} \left(M^{\prime}_2-M_2\right) \geq$ $2 \delta$. If $M^{\prime}_3=M_3$ and $r \neq r^{\prime}$, then we have $d_S\left(\mathcal{U}_1, \mathcal{U}_2\right) \geq 2 \left(\mathrm{rank}\left(M^{\prime}_1-M_1\right)+\mathrm{rank}\left(M^{\prime}_2-M_2\right)\right) \geq$ $2\left(b_1+b_2\right) \geq 2 \delta$. The conclusion is proved.
	\end{proof}
	
	\begin{lem}\label{lem16}  Let $n_1$, $n_2$, $n_3$, $\delta_1$, $\delta_2$, $b_1$, $b_2$, $k_3$, $k_4$ be positive integers satisfying $n_1+n_2+n_3=n$, $\delta_1+\delta_2=k$, $b_1+b_2 \geq \delta$ and $1<b_i<\delta$ for $i=1,2$. 
		Presume that	$\mathcal{D}_3$ is a linear $\left[ \delta_1 \times n_3, k_3,b_1\right]_q$ RMC, $\mathcal{M}_3$ is a linear $\left[ \delta_1 \times n_3, {k}_3^{\prime}, \delta\right]_q$ subcode of $\mathcal{D}_3$.
	Assume that  $\mathcal{D}_4$ is a linear  $\left[ \delta_2 \times (n_2-\delta_2),k_4, b_2\right] _q$ RMC, $\mathcal{M}_4$ is a linear $\left[ \delta_2 \times (n_2-\delta_2), {k}_4^{\prime},\delta\right] _q$ subcode of $\mathcal{D}_4$.
	Then $s=\min\{{q}^{k_3-{k}_3^{\prime}},{q}^{k_4-{k}_4^{\prime}}\}$ RMCs satisfying the following conditions can be constructed:
	\begin{itemize}
		\item $\mathcal{M}_1^u \left(\delta_1 \times n_3,\left|\mathcal{M}_1^u\right|, \delta\right)_q$ and 
		$\mathcal{M}_2^u \left(\delta_2 \times (n_2-\delta_2),\left|\mathcal{M}_2^u\right|, \delta\right)_q$ are RMCs for all  $1 \leq u\leq s$;
		\item For $M \in \mathcal{M}_i^u$ and $M^{\prime} \in \mathcal{M}_i^{u^{\prime}}$ satisfy $M \neq M^{\prime}$ and ${rank}\left(M-M^{\prime}\right) \geq b_i$ for all $1\leq i\leq 2$, $1 \leq u<u^{\prime} \leq s $. 
	\end{itemize}	
		Let $\mathcal{M}_3$ is a $\left(\delta_2 \times n_3,\left|\mathcal{M}_3\right|, \delta\right)_q$ RMC. 
		 Then $\mathcal{C}_5=\bigcup_{u=1}^s \mathcal{C}_u$ is an $(n, *, 2 \delta, k)_q$-CDC, where 
		$$
		\mathcal{C}_u =\left\{{rs}\left[\begin{array}{ccccc}
			0_{\delta_2 \times\left(n_1 - \delta_1\right)} &0_{\delta_2 \times \delta_1} & M_3 & M_2 & \hat{I}_{\delta_2}\\
			0_{\delta_1 \times\left(n_1 - \delta_1\right)}& I_{\delta_1} & M_1 & 0_{\delta_1 \times\left(n_2-\delta_2\right)}& 0_{\delta_1 \times \delta_2} 
		\end{array}\right]: M_1 \in \mathcal{M}_1^u, M_2 \in \mathcal{M}_2^u, M_3 \in \mathcal{M}_3\right\}.
		$$
	\end{lem}
	\begin{proof}
		The proof is similar to the Lemma 15 and we omit it here.
	\end{proof}
	
Based on Lemma 15 and Lemma 16, we insert the CDCs constructed by inverse bilateral multilevel construction into $\mathcal{C}_1 \cup \mathcal{C}_2 \cup \mathcal{C}_3$ by the double multilevel construction and bilateral multilevel construction, obtaining some CDCs with larger size than the known codes.
	
	\begin{theo}\label{theo2} For any positive integer $\delta$, it holds that
		$$
		\begin{aligned}
			\bar{A}_q(6 \delta, 2 \delta, 3 \delta)
			\geq q^{6 \delta^2+3 \delta}+1+\sum_{i=\delta}^{2 \delta} a(q, 3 \delta, 3 \delta, \delta, i)+q^{2 \delta^2+4 \delta+\delta\lceil\frac{\delta}{2}\rceil}+q^{\delta^2+5 \delta} 
			+q^{5\delta+\delta\lceil\frac{\delta}{2}\rceil}+q^{4 \delta+\delta\lceil\frac{\delta}{2}\rceil}.
		\end{aligned}
		$$
	\end{theo}
	
	\begin{proof} We choose an identifying vector of length $6 \delta$ as follows: $v=(\underbrace{1 \cdots 1}_{3 \delta} \underbrace{0 \cdots 0}_{3 \delta})$
		and inverse identifying vector of length $6 \delta$ as follows: $\hat{v}=(\underbrace{0 \cdots 0}_{3 \delta} \underbrace{1 \cdots 1}_{3 \delta}).$ We also choose bilateral identifying vectors of length $6 \delta$ as follows:
		$$
		\begin{gathered}
			\bar{v}_1=(\underbrace{1 \cdots 1}_{2 \delta} \underbrace{0 \cdots 0}_\delta \underbrace{\bar{0} \cdots \bar{0}}_\delta \underbrace{\hat{0} \cdots \hat{0}}_\delta \underbrace{\hat{1} \cdots \hat{1}}_\delta),
		\end{gathered}
		$$
		
		$$
		\begin{gathered}
			\bar{v}_2=(\underbrace{1 \cdots 1}_\delta \underbrace{0 \cdots 0}_\delta \underbrace{1 \cdots 1}_\delta \underbrace{\bar{0} \cdots \bar{0}}_\delta \underbrace{\hat{0} \cdots \hat{0}}_\delta \underbrace{\hat{1} \cdots \hat{1}}_\delta),
		\end{gathered}
		$$
		and inverse bilateral identifying vectors of length $6 \delta$ as follows:
		$$
		\begin{aligned}
			& \hat{\bar{v}}_1=(\underbrace{\tilde{1} \cdots \tilde{1}}_{2 \delta} \underbrace{\tilde{0} \cdots \tilde{0}}_\delta \underbrace{\bar{\tilde{0}} \cdots \bar{\tilde{0}}}_\delta \underbrace{\hat{\tilde{1}} \cdots \hat{\tilde{1}}}_\delta \underbrace{\hat{\tilde{0}} \cdots \hat{\tilde{0}}}_\delta), 
		\end{aligned}
		$$
		$$
		\begin{aligned}
			& \hat{\bar{v}}_2=(\underbrace{\tilde{0} \cdots \tilde{0}}_\delta \underbrace{\tilde{1} \cdots \tilde{1}}_{2\delta} \underbrace{\bar{\tilde{0}} \cdots \bar{\tilde{0}}}_\delta  \underbrace{\hat{\tilde{0}} \cdots \hat{\tilde{0}}}_\delta \underbrace{\hat{\tilde{1}} \cdots \hat{\tilde{1}}}_\delta) . 
			&
		\end{aligned}
		$$
		
		Let $\mathcal{C}_1$ be a $(6\delta, 2\delta, 3\delta)_q$-CDC with identifying vector $v$, $\mathcal{C}_2$ be a $(6\delta, 2\delta, 3\delta)_q$-CDC with inverse identifying vector $\hat{v}$ and $\mathcal{C}_3$ be a $(6\delta, 2\delta, 3\delta)_q$-CDC with  bilateral identifying vectors $\bar{v}_1$ and $\bar{v}_2$. By \cite{30} we have
		$\mathcal{C}_1\cup\mathcal{C}_2\cup\mathcal{C}_3$ is a $(6\delta, 2\delta, 3\delta)_q$-CDC, where 
		$|\mathcal{C}_1\cup\mathcal{C}_2\cup\mathcal{C}_3|\geq q^{6 \delta^2+3 \delta}+1+\sum_{i=\delta}^{2 \delta} a(q, 3 \delta, 3 \delta, \delta, i)+q^{2 \delta^2+4 \delta+\delta\lceil\frac{\delta}{2}\rceil}+q^{\delta^2+5 \delta}$.
		
		The inverse bilateral echelon Ferrers form of $\hat{\bar{v}}_1$ is
		$$
		\widehat{\overline{EF}}\left(\hat{\bar{v}}_1\right)=\left[\begin{array}{ccccc}
			0 & 0 & \mathcal{F}_2 & \hat{I}_\delta & 0\\
			I_{2 \delta} & \mathcal{F}_1 & \mathcal{F}_3 & 0 & 0 
		\end{array}\right],
		$$
		where $\mathcal{F}_1$, $\mathcal{F}_3$ are $2 \delta \times \delta$ full Ferrers diagrams and $\mathcal{F}_2$ is a $\delta \times \delta$ full Ferrers diagram. By Lemma 15, we assume that $n_1=3 \delta$, $n_2=2 \delta$, $n_3=\delta$, $\delta_1=2 \delta$, $\delta_2=\delta$, $b_1=\lceil\frac{\delta}{2}\rceil$, $b_2=\lfloor\frac{\delta}{2}\rfloor$ and
		$$
		\widehat{\overline{\mathcal{C}}}_r=\left\{\mathrm{ rs }\left[\begin{array}{ccccc}
			0 & 0 & {M}_2 & \hat{I}_\delta & 0\\
			I_{2 \delta} & {M}_1 & {M}_3 & 0 & 0 
		\end{array}\right]: {M}_1 \in \mathcal{M}_1^r, M_2 \in \mathcal{M}_2^r, M_3 \in \mathcal{M}_3\right\}, 
		$$
		where $\mathcal{M}_1^r$ is a $[2 \delta \times \delta, \delta]_q$-MRD subcode of a $\left[2 \delta \times \delta, b_1\right]_q$-MRD code,  $\mathcal{M}_2^r$ is a $[\delta \times \delta, \delta]_q$-MRD subcode of a $\left[\delta \times \delta, b_2\right]_q$-MRD code for all $1 \leq r \leq s=$ $q^{\delta\lceil\frac{\delta}{2}\rceil}$ and $\mathcal{M}_3$ is a $[2 \delta \times \delta, \delta]_q$-MRD code. It can be seen from Lemma 14 that $\mathcal{M}_1^r$ and $\mathcal{M}_2^r$ satisfy the conditions in Lemma 15. Then by Lemma 15, we have $|\mathcal{C}_4|=|\bigcup_{r=1}^s \widehat{\overline{\mathcal{C}}}_r|=s \cdot q^{2 \delta} \cdot q^{2 \delta} \cdot q^\delta=q^{5 \delta+\delta\lceil\frac{\delta}{2}\rceil}$.

		The inverse bilateral echelon Ferrers form of $\hat{\bar{v}}_2$ is given by
		$$
		\widehat{\overline{EF}}\left(\hat{\bar{v}}_2\right)=\left[\begin{array}{cccccc}
			0 & 0 & \mathcal{F}_3 & \mathcal{F}_2 & \hat{I}_\delta\\
			0 &I_{2 \delta} & \mathcal{F}_1 & 0 & 0 
		\end{array}\right],
		$$
		where $\mathcal{F}_1$ is a $2\delta \times \delta$ full Ferrers diagram and $\mathcal{F}_2$, $\mathcal{F}_3$ are $ \delta \times \delta$ full Ferrers diagrams. By Lemma 16, we assume that $n_1=3 \delta$, $n_2=2 \delta$, $n_3=\delta$, $\delta_1=2\delta$, $\delta_2=\delta$, $b_1=\lceil\frac{\delta}{2}\rceil$, $b_2=$ $\lfloor\frac{\delta}{2}\rfloor$ and
		
		$$\widehat{\overline{\mathcal{C}}}_u=\left\{\mathrm{rs} \left[\begin{array}{cccccc}
			0 &0 & M_3 & M_2 & \hat{I}_{\delta}\\
			0& I_{2\delta} & M_1 & 0& 0
		\end{array}\right]: {M}_1 \in \mathcal{M}_1^u, {M}_2 \in \mathcal{M}_2^u, M_3 \in \mathcal{M}_3\right\},
		$$where $\mathcal{M}_1^u$ is a $[2\delta \times \delta, \delta]_q$-MRD subcode of a $\left[2\delta \times \delta, b_1\right]_q$-MRD code, $\mathcal{M}_2^u$ is a $[\delta \times \delta, \delta]_q$-MRD subcode of a $\left[\delta \times \delta, b_2\right]_q$-MRD code for all $1 \leq u \leq$ $s=q^{\delta\lceil\frac{\delta}{2}\rceil}$ and $\mathcal{M}_3$ is a $[\delta \times \delta, \delta]_q$-MRD code. It can be seen from Lemma 14 that $\mathcal{M}_1^u$ and $\mathcal{M}_2^u$ satisfy the conditions in Lemma 16. Then by Lemma 16, we have $|\mathcal{C}_5|=|\bigcup_{u=1}^s \widehat{\overline{\mathcal{C}}}_u|=s \cdot q^{2 \delta} \cdot q^\delta \cdot q^\delta=q^{4 \delta+\delta\lceil\frac{\delta}{2}\rceil}$. 
		
		Let $\mathcal{C}=\bigcup_{i=1}^5 \mathcal{C}_i$, then 
		it is sufficient to examine the subspace distance of $\mathcal{C}$. It is easy to see that $\mathcal{C}_4\cup \mathcal{C}_5$ is a $(6\delta,q^{5 \delta+\delta\lceil\frac{\delta}{2}\rceil}+q^{4 \delta+\delta\lceil\frac{\delta}{2}\rceil}, 2\delta, 3\delta)_q$-CDC constructed by inverse bilateral multilevel construction. We just need to verify $d_S(\mathcal{C}_i, \mathcal{C}_4) \geq 2\delta $ and $d_S(\mathcal{C}_i, \mathcal{C}_5) \geq 2\delta $ for $i=1, 2, 3$.
		
		For any $$\left\{\mathrm{ rs }\left[\begin{array}{ccccc}
			I_{2\delta} & 0 & {A}_{11} & A_{12} & A_{13}\\
			0 & I_\delta & A_{21} & A_{22} & A_{23}
		\end{array}\right] \right\} \in \mathcal{C}_1\  \mathrm{and} \ 
		\left\{\mathrm{ rs }\left[\begin{array}{ccccc}
			0 & 0 & M_2 & \hat{I}_\delta& 0\\
			I_{2\delta} & M_1& M_3 & 0 & 0
		\end{array}\right] \right\} \in \mathcal{C}_4, $$ 
		where $ \left[\begin{array}{ccc}
			 {A}_{11} & A_{12} & A_{13}\\
			 A_{21} & A_{22} & A_{23}
		\end{array}\right]$ is a codeword of a $[3\delta \times 3\delta, \delta]_q$-MRD code,
		then we have
		
		\begin{align*}
			d_S(\mathcal{C}_1, \mathcal{C}_4)
			&=2\mathrm{rank}\left[\begin{array}{ccccc}
				I_{2\delta} & 0 & {A}_{11} & A_{12} & A_{13}\\
				0 & I_\delta & A_{21} & A_{22} & A_{23}\\
				0 & 0 & M_2 & \hat{I}_\delta& 0\\
				I_{2\delta} & M_1& M_3 & 0 & 0
			\end{array}\right]-6\delta\\
			& =2\mathrm{rank}\left[\begin{array}{ccccc}
				I_{2\delta} & 0 & 0 & 0 & 0\\
				0 & I_\delta & A_{21} & A_{22} & A_{23}\\
				0 & 0 & M_2 & \hat{I}_\delta& 0\\
				0& M_1& M_3-A_{11} & -A_{12} & -A_{13}
			\end{array}\right]-6\delta \\
			&=2\mathrm{rank}
			\left[\begin{array}{ccccc}
				I_{2\delta} & 0 & 0 & 0 & 0\\
				0 & I_\delta & A_{22} & A_{21} & A_{23}\\
				0& M_1& -A_{12}& M_3-A_{11}  & -A_{13}\\
				0 & 0 & \hat{I}_\delta& M_2 & 0
			\end{array}\right]-6\delta \geq 2\delta.
		\end{align*} 
		
		Similarly, it is easy to see that $d_S(\mathcal{C}_i, \mathcal{C}_4) \geq 2\delta$, $i=2, 3$ and $d_S(\mathcal{C}_i, \mathcal{C}_5) \geq 2\delta$ for $i=1, 2, 3.$ Then
		$\mathcal{C}$ is a $(6 \delta, M, 2 \delta, 3 \delta)_q$-CDC, where $$M=q^{6 \delta^2+3 \delta}+1+\sum_{i=\delta}^{2 \delta} a(q, 3 \delta, 3 \delta, \delta, i)+q^{2 \delta^2+4 \delta+\delta\lceil\frac{\delta}{2}\rceil}+q^{\delta^2+5 \delta}+q^{5 \delta+\delta\lceil\frac{\delta}{2}\rceil}+q^{4 \delta+\delta\lceil\frac{\delta}{2}\rceil}.$$
	\end{proof}
	
	\begin{coro}\label{coro2} Set $\delta=3$ in Theorem 2, we have
		$$
		\bar{A}_q(18,6,9) \geq q^{63}+\sum_{i=3}^6 a(q, 9,9,3, i)+q^{36}+q^{24}+q^{21}+q^{18}+1,
		$$
		which is larger than the lower bound of $q^{63}+\sum_{i=3}^6 a(q, 9,9,3, i)+q^{36}+q^{24}+1$ (see \cite{30}).
	\end{coro}
	
	Applying Theorem 2 with Corollary 2, we improved the lower bounds of the size for some CDCs by inserting the inverse bilateral multilevel construction into the double multilevel construction and bilateral multilevel construction. Below, by introducing a new set of bilateral identifying vectors, we aim to further improve the lower bounds for CDCs with other parameters.
	
	\begin{theo}\label{theo3} Let $n \geq 2 k \geq 2 \delta$. Suppose that $\mathcal{C}_1$ and $\mathcal{C}_2$ are CDCs in Lemma 10 with the identifying vector set $\mathcal{A}$ and inverse identifying vector set $\hat{\mathcal{A}}$, respectively. Assume that
		$\mathcal{C}_3$ is a CDC in Lemma 12 with the bialteral identifying vector set $\mathcal{B}$ and $\mathcal{C}_6$ is an $\left(n, |\mathcal{C}_6|, 2 \delta, k\right)_q$-CDC
		 constructed by the bilateral multilevel construction with bilateral identifying vector set $\mathcal{B}_1$. If the identifying vector $u \in \mathcal{A}$,  inverse identifying vector $\hat{u} \in \hat{\mathcal{A}}$,
		 bilateral identifying vector $\bar{v}=$ $(\overbrace{v_1}^{n_1}|\overbrace{\bar{v}_3}^{n-n_1-n_2}| \overbrace{\hat{v}_2}^{n_2}) \in \mathcal{B}$
		and bilateral identifying vector ${\bar{v}^{\prime}}=$ $(\overbrace{v_1^{\prime}}^{n_1}|\overbrace{\bar{v}_3^{\prime}}^{n-n_1-n_2}| \overbrace{\hat{v}_2^{\prime}}^{n_2})\in \mathcal{B}_1$ satisfy the following conditions:\\
		(1) the inequality in the double multilevel construction holds;\\
		(2) $d_H\left(v_1, u_1\right)+\left|w t\left(v_1\right)-w t\left(u_1\right)\right| \geq 2 \delta$ and
		$d_H\left(\hat{v}_2, \hat{u}_2\right)+\left|w t\left(\hat{v}_2\right)-w t\left(\hat{u}_2\right)\right| \geq 2 \delta$;\\
		(3) $d_H\left(v^{\prime}_1, u_1\right)+\left|w t\left(v^{\prime}_1\right)-w t\left(u_1\right)\right| \geq 2 \delta$ and
		$d_H\left(\hat{v}^{\prime}_2, \hat{u}_2\right)+\left|w t\left(\hat{v}^{\prime}_2\right)-w t\left(\hat{u}_2\right)\right| \geq 2 \delta$;\\
		(4) $d_H\left(v_1, v_1^{\prime}\right)+d_H\left(\hat{v}_2, \hat{v}_2^{\prime}\right) \geq 2 \delta$,\\ where $u_1$ is a subvector consisting of the first $n_1$ coordinates of $u$ and $\hat{u}_2$ is a subvector consisting of the last $n_2$ coordinates of $\hat{u}$, then $\mathcal{C}_1 \cup \mathcal{C}_2 \cup \mathcal{C}_3 \cup \mathcal{C}_6$ is an $\left(n, M_1+M_2+|\mathcal{C}_3|+|\mathcal{C}_6|, 2 \delta, k\right)_q$-CDC.
	\end{theo}

	\begin{proof}
		
		It is easy to obtain that $\mathcal{C}_1 \cup \mathcal{C}_2 \cup \mathcal{C}_3$ is an $\left(n, M_1+M_2+|\mathcal{C}_3|, 2 \delta, k\right)_q$-CDC by \cite{30}. We only need to prove that $d_S(\mathcal{C}_i, \mathcal{C}_6)\geq 2\delta$ for $i=1,2,3$.
		
		Clearly, for any $\mathcal{U} \in \mathcal{C}_1$ and $\mathcal{V} \in \mathcal{C}_6$,  by condition (3) we have
		$
		d_S(\mathcal{U}, \mathcal{V}) \geq 2 \delta.
		$
		Similarly, for any $\mathcal{U} \in \mathcal{C}_2$ and $\mathcal{V} \in \mathcal{C}_6$, we have $d_S(\mathcal{U}, \mathcal{V}) \geq 2 \delta$.
		
For any $\mathcal{U} \in \mathcal{C}_3$ and $\mathcal{V} \in \mathcal{C}_6$, suppose the bilateral identifying 
vector of $\mathcal{U}$ is   $\bar{v}=(\overbrace{v_1}^{n_1}|\overbrace{\bar{v}_3}^{n-n_1-n_2}| \\ \overbrace{\hat{v}_2}^{n_2}) \in \mathcal{B}$ and bilateral identifying vector of $\mathcal{V}$ is ${\bar{v}^{\prime}}=(\overbrace{v_1^{\prime}}^{n_1}|\overbrace{\bar{v}_3^{\prime}}^{n-n_1-n_2}| \overbrace{\hat{v}_2^{\prime}}^{n_2}) \in {\mathcal{B}_1}$. It can be seen that they have the same type, then we have $d_S(\mathcal{U}, \mathcal{V}) \geq d_H(\bar{v}, {\bar{v}^{\prime}}) = d_H\left(v_1, v_1^{\prime}\right)+d_H\left(\hat{v}_2, \hat{v}_2^{\prime}\right) \geq 2 \delta$ by condition (4).
	\end{proof}
	In Theorem 3, the subspace distance only relies on the Hamming distance of identifying vectors, inverse identifying vectors and bilateral identifying vectors. Therefore, we have the following result.
	
	\begin{coro}\label{coro3}Let $n, k, \delta$ be positive integers such that $n \geq 2k+\delta$ with $\lfloor\frac{k}{2}\rfloor \geq$ $2\lceil\frac{\delta}{2}\rceil$ and $ t= min\{\lfloor\frac{\lceil\frac{k}{2}\rceil}{\lceil\frac{\delta}{2}\rceil}\rfloor, \lfloor\frac{\lfloor\frac{k}{2}\rfloor}{\lfloor\frac{\delta}{2}\rfloor}  \rfloor \}$.
		Then there exists an $(n, M^{\prime}, 2 \delta, k)_q$-CDC, where
		\begin{align*}
			&M' =  q^{(n-k)(k-\delta+1)} \frac{1-q^{-\lfloor\frac{k}{\delta}\rfloor \delta^2}}{1-q^{-\delta^2}}+q^{(n-k-\delta)(k-\delta+1)}+1+\sum_{i=\delta}^{k-\delta} a(q, k, n-k, \delta, i) \\
			&~~~~~~+t q^{(n-k-2 \delta+\lfloor\frac{\delta}{2}\rfloor)(\lceil\frac{k}{2}\rceil-\delta+1)}+
			(t-1)q^{(n-k-2\delta+\lfloor\frac{\delta}{2}\rfloor-\lceil\frac{\delta}{2}\rceil)(\lceil\frac{k}{2}\rceil-\delta+1)}.
		\end{align*}
		In particular, the CDC contains a lifted MRD code as a subcode.
	\end{coro}
	
	\begin{proof}
		Let $\mathcal{C}_1$ be an $\left(n, M_1, 2 \delta, k\right)_q$-CDC constructed via the multilevel construction, whose set of identifying vectors is below:
		$$
		\mathcal{A}=\{(\underbrace{1 \cdots 1}_k 0 \cdots 0)\}\cup\{(\underbrace{1 \cdots 1}_{k-i \delta} \underbrace{0 \cdots 0}_\delta \underbrace{1 \cdots 1}_{i \delta} 0 \cdots 0) \mid 1 \leq i \leq\lfloor\frac{k}{\delta}\rfloor\} .
		$$
		Then $M_1=q^{(n-k)(k-\delta+1)} \frac{1-q^{-\lfloor\frac{k}{\delta}\rfloor \delta^2}}{1-q^{-\delta^2}}+q^{(n-k-\delta)(k-\delta+1)}$ by Lemma 13.
		
		Assume $\mathcal{C}_2=\left\{\operatorname{rs}\left(A \mid I_k\right): A \in \mathcal{M}, \operatorname{rank}(A) \leq k-\delta\right\}$, where $\mathcal{M}$ is a $[k \times(n-$ $k), \delta]_q$-MRD code. By Lemma 5 and Lemma 3, $\mathcal{C}_2$ is an $\left(n, M_2, 2 \delta, k\right)_q$-CDC, where $$M_2 \geq A_q^G(k \times(n-k), \delta,[0, k-\delta]) \geq 1+\sum_{i=\delta}^{k-\delta} a(q, k, n-k, \delta, i).$$
		
		Suppose $\mathcal{C}_3$ is an $\left(n, |\mathcal{C}_3|, 2 \delta, k\right)_q$-CDC constructed via the bilateral multilevel construction, whose set of bilateral identifying vectors is below:
		\begin{align*}
			\mathcal{B}=\{\bar{v}_j=(\overbrace{\underbrace{1\ldots1}_{\lceil\frac{k}{2}\rceil-j\lceil\frac{\delta}{2}\rceil} \underbrace{0\ldots0}_{\lceil\frac{\delta}{2}\rceil}\underbrace{1\ldots1}_{j\lceil\frac{\delta}{2}\rceil}0\ldots0}^{k+\delta}\bar{0}\ldots\bar{0}
			\underbrace{\hat{0}\ldots\hat{0}}_{\lceil\frac{\delta}{2}\rceil} \underbrace{\hat{1}\ldots\hat{1}}_{j\lfloor\frac{\delta}{2}\rfloor} \underbrace{\hat{0}\ldots\hat{0}}_{\lfloor\frac{\delta}{2}\rfloor} \underbrace{\hat{1}\ldots\hat{1}}_{\lfloor\frac{k}{2}\rfloor-j\lfloor\frac{\delta}{2}\rfloor}):0\leq j< t\}.
		\end{align*}
		According to Corollary 1 in \cite{30}, we have $|\mathcal{C}_3| \geq t q^{(n-k-2 \delta+\lfloor\frac{\delta}{2}\rfloor)( \lceil\frac{k}{2}\rceil-\delta+1)}$.
		
		Let $\mathcal{C}_6$ be an $\left(n, |\mathcal{C}_6|, 2 \delta, k\right)_q$-CDC constructed via the bilateral multilevel construction, whose set of bilateral identifying vectors is below:
		\begin{align*}
			\mathcal{B}_1=\{{\bar{v}'}_l=(\overbrace{\underbrace{0\ldots0}_{\lceil\frac{\delta}{2}\rceil}\underbrace{1\ldots1}_{\lceil\frac{k}{2}\rceil-l\lceil\frac{\delta}{2}\rceil} \underbrace{0\ldots0}_{\lceil\frac{\delta}{2}\rceil}\underbrace{1\ldots1}_{l\lceil\frac{\delta}{2}\rceil}0\ldots0}^{k+\delta}\bar{0}\ldots\bar{0}
			\underbrace{\hat{1}\ldots\hat{1}}_{l\lfloor\frac{\delta}{2}\rfloor} \underbrace{\hat{0}\ldots\hat{0}}_{\lfloor\frac{\delta}{2}\rfloor} \underbrace{\hat{1}\ldots\hat{1}}_{\lfloor\frac{k}{2}\rfloor-l\lfloor\frac{\delta}{2}\rfloor}\underbrace{\hat{0}\ldots\hat{0}}_{\lceil\frac{\delta}{2}\rceil} ):1\leq l< t\}.
		\end{align*}

		For any bilateral identifying vector ${\bar{v}'}_l$, its bilateral echelon Ferrers form is
		$$
		\left[\begin{array}{cccccccccc}
			0 & I_{\lceil\frac{k}{2}\rceil-l\lceil\frac{\delta}{2}\rceil} & \mathcal{F}_1 & 0 & \mathcal{F}_2 & \mathcal{F}_3 & 0 & 0 & 0 & 0 \\
			0 & 0 & 0 & I_{l\lceil\frac{\delta}{2}\rceil} & \mathcal{F}_4 & \mathcal{F}_5 & 0 & 0 & 0 & 0 \\
			0 & 0 & 0 & 0 & 0 & \mathcal{F}_6 & 0 & \mathcal{F}_7 & I_{\lfloor\frac{k}{2}\rfloor-l\lfloor\frac{\delta}{2}\rfloor} & 0 \\
			0 & 0 & 0 & 0 & 0 & \mathcal{F}_8 & I_{l\lfloor\frac{\delta}{2}\rfloor} & 0 & 0 & 0
		\end{array}\right],
		$$
		where $\mathcal{F}_1$ is a $(\lceil\frac{k}{2}\rceil-l\lceil\frac{\delta}{2}\rceil)\times \lceil\frac{\delta}{2}\rceil$ full Ferrers diagram,
		$\mathcal{F}_2$ is a $(\lceil\frac{k}{2}\rceil-l\lceil\frac{\delta}{2}\rceil)\times(\lfloor\frac{k}{2}\rfloor+\lfloor\frac{\delta}{2}\rfloor-\lceil\frac{\delta}{2}\rceil) $ full Ferrers diagram,
		$\mathcal{F}_3$ is a $ (\lceil\frac{k}{2}\rceil-l\lceil\frac{\delta}{2}\rceil)\times(n-k-2\delta-\lfloor\frac{k}{2}\rfloor)$ full Ferrers diagram,
		$\mathcal{F}_4$ is a $ l\lceil\frac{\delta}{2}\rceil\times(\lfloor\frac{k}{2}\rfloor+\lfloor\frac{\delta}{2}\rfloor-\lceil\frac{\delta}{2}\rceil)$ full Ferrers diagram,
		$\mathcal{F}_5$ is a $ l\lceil\frac{\delta}{2}\rceil\times(n-k-2\delta-\lfloor\frac{k}{2}\rfloor)$ full Ferrers diagram,
		$\mathcal{F}_6$ is a $ (\lfloor\frac{k}{2}\rfloor-l\lfloor\frac{\delta}{2}\rfloor)\times (n-k-2\delta-\lfloor\frac{k}{2}\rfloor)  $ full Ferrers diagram,
		$\mathcal{F}_7$ is a $(\lfloor\frac{k}{2}\rfloor-l\lfloor\frac{\delta}{2}\rfloor)\times\lfloor\frac{\delta}{2}\rfloor $ full Ferrers diagram and
		$\mathcal{F}_8$ is a $ l\lfloor\frac{\delta}{2}\rfloor\times(n-k-2\delta-\lfloor\frac{k}{2}\rfloor) $ full Ferrers diagram.
		
		Denote
		$$
		\mathcal{F}=\begin{array}{cc}
			\mathcal{F}_2 & \mathcal{F}_3 \\
			\mathcal{F}_4 & \mathcal{F}_5
		\end{array},
		$$
		then $\mathcal{F}$ is a $ \lceil\frac{k}{2}\rceil\times(n-k-2\delta+\lfloor\frac{\delta}{2}\rfloor-\lceil\frac{\delta}{2}\rceil) $ full Ferrers diagram.
		By $\lfloor\frac{k}{2}\rfloor \geq 2\lceil\frac{\delta}{2}\rceil$ , we have $ n-k-2\delta+\lfloor\frac{\delta}{2}\rfloor-\lceil\frac{\delta}{2}\rceil\geq k-2\lceil\frac{\delta}{2}\rceil\geq \lceil\frac{k}{2}\rceil $, we can fill in $\mathcal{F}$ with a $[\lceil\frac{k}{2}\rceil\times(n-k-2\delta+\lfloor\frac{\delta}{2}\rfloor-\lceil\frac{\delta}{2}\rceil), \delta]_q$-MRD code.
		Thus $|\mathcal{C}_6|\geq(t-1)q^{(n-k-2\delta+\lfloor\frac{\delta}{2}\rfloor-\lceil\frac{\delta}{2}\rceil)\times(\lceil\frac{k}{2}\rceil-\delta+1)}$.
		
		Let $\mathcal{C}'=\mathcal{C}_1 \cup \mathcal{C}_2 \cup \mathcal{C}_3 \cup \mathcal{C}_6$. We claim that $\mathcal{C}'$ is an $(n, M', 2 \delta, k)_q$-CDC. It suffices to examine the subspace distance of $\mathcal{C}'$.
		
		Note that 
		$\operatorname{wt}\left(v_1^{\prime}\right)=\lceil\frac{k}{2}\rceil$ and $\operatorname{wt}\left(\hat{v}_2^{\prime}\right)=\lfloor\frac{k}{2}\rfloor$ for any bilateral identifying vector $${\bar{v}'}=(\overbrace{v_1^{\prime}}^{n_1}|\overbrace{\bar{v}_3^{\prime}}^{n-n_1-n_2}|  \overbrace{\hat{v}_2^{\prime}}^{n_2}) \in {\mathcal{B}_1}.$$
		For any two bilateral identifying vectors ${\bar{u}'}$, ${\bar{v}'}$ $\in {\mathcal{B}_1}$ and ${\bar{u}'} \neq {\bar{v}'}$, we have
		$$
		d_H({\bar{u}'}, {\bar{v}'}) \geq 2\lceil\frac{\delta}{2}\rceil+2\lfloor\frac{\delta}{2}\rfloor=2 \delta .       
		$$
		For any identifying vector $u \in \mathcal{A}$, let $u_1$ be the subvector consisting of the first $n_1$ coordinates of $u$. For any bilateral identifying vector ${\bar{v}^{\prime}}=(\overbrace{\bar{v}_1^{\prime}}^{n_1}|\overbrace{\bar{v}_3^{\prime}}^{n-n_1-n_2}| \overbrace{\hat{v}_2^{\prime}}^{n_2}) \in {\mathcal{B}_1}$, where $n_1=k+\delta$ and $n_2=\lfloor\frac{k}{2}\rfloor+\delta$, it is easy to obtain that
		\begin{align*}
			d_H(v^{\prime}_1, u_1)+|\operatorname{wt}\left(v^{\prime}_1)-
			\operatorname{wt}(u_1\right)| \geq 2|\operatorname{wt}\left(v^{\prime}_1)-\operatorname{wt}(u_1\right)| \geq2 \delta.
		\end{align*}
		For the inverse identifying vector $\hat{u}=(0 \cdots 0 \underbrace{1 \cdots 1}_k)$, let $\hat{u}_2$ be the subvector consisting of the last $n_2$ coordinates of $\hat{u}$ and it holds that
		\begin{align*}
			d_H(\hat{v}^{\prime}_2, \hat{u}_2)+\left|\operatorname{w t}\left(\hat{v}^{\prime}_2\right)-\operatorname{w t}\left(\hat{u}_2\right)\right|
			\geq 2\left|\operatorname{w t}\left(\hat{v}^{\prime}_2\right)-\operatorname{w t}\left(\hat{u}_2\right)\right|
			=2 \delta.
		\end{align*}
		For any $\mathcal{U} \in \mathcal{C}_3$ and  $\mathcal{V} \in \mathcal{C}_6$, it is easy to see that 
		\begin{align*}
			d_S(\mathcal{U}, \mathcal{V}) \geq d_H(\bar{v}, {\bar{v}^{\prime}})=d_H\left(v_1, v_1^{\prime}\right)+d_H\left(\hat{v}_2, \hat{v}_2^{\prime}\right) \geq 2\lceil\frac{\delta}{2}\rceil+2\lfloor\frac{\delta}{2}\rfloor=2\delta .
		\end{align*}
		Applying Theorem 3, $\mathcal{C}'$ is the desired CDC which contains a lifted MRD code as a subset.
	\end{proof}
	
	\begin{rem}\label{rem1}
		Notice that in Corollary 2, we have dropped many dots. We can get better lower bounds in most cases. In the following corollaries, we provide better lower bounds.
	\end{rem}
	
	\begin{coro}\label{coro4}
		Let $n=17$, $k=6$, $\delta=3$. Then
		$
		\bar{A}_q(17,6,6) \geq q^{44}+q^{35}+q^{32}+(q^4+q^3+q^2+q+1)(q^3+1)(q^2+1)(q^{11}-1)+q^{14}+q^{12}+q^{8}.
		$ 
	\end{coro}
	\begin{proof}
		We construct a set of identifying vectors as follows:
		\begin{align*}
			\mathcal{A}=\{v_1=(111111 \underbrace{0 \cdots 0}_{11}), v_2=(111000111 \underbrace{0 \cdots 0}_8), v_3=(000111111 \underbrace{0 \cdots 0}_8)\},
		\end{align*}
		a set of inverse identifying vectors as follows:
		\begin{align*}
			\hat{\mathcal{A}}=\{(\hat{v}=(\underbrace{0 \cdots 0}_{11} 111111)\},
		\end{align*}
		and two sets of bilateral identifying vectors as follows:
		\begin{align*}
			\mathcal{B}=\{&\bar{v}_1=(111000000 \bar{0} \bar{0}  \hat{0} \hat{0} \hat{0} \hat{1} \hat{1} \hat{1}), \bar{v}_2=(100110000 \bar{0}\bar{0} \hat{0} \hat{0} \hat{1} \hat{0} \hat{1} \hat{1})\},
		\end{align*}
		\begin{align*}
			{\mathcal{B}_1}=\{\bar{v}_1^{\prime}=(001001100 \bar{0}\bar{0} \hat{1} \hat{0} \hat{1} \hat{1} \hat{0} \hat{0})\}.
		\end{align*}
		By Corollary 2, we have the sizes of CDCs corresponding to $v_1$, $v_2$, $v_3$ and $\hat{v}$ are $q^{44}$, $q^{35}$, $q^{32}$ and $1+\sum_{i=3}^3 a(q, 6,11,3, i)$, respectively. The corresponding CDCs for $\bar{v}_1$ and $\bar{v}_2$ have the sizes $q^{14}$ and $q^{12}$, respectively by \cite{30}.
		
		The bilateral echelon Ferrers form of ${\bar{v}}_1^{\prime}$ is
		$$
		\left[\begin{array}{ccccccccccccccccc}
			0&0&\mathbf{1}&\bullet&\bullet&0&0&\bullet&\bullet&\bullet&\bullet&0&0&0&0&0&0\\
			0&0&0&0&0&\mathbf{1}&0&\bullet&\bullet&\bullet&\bullet&0&0&0&0&0&0\\
			0&0&0&0&0&0&\mathbf{1}&\bullet&\bullet&\bullet&\bullet&0&0&0&0&0&0\\
			0&0&0&0&0&0&0&0&0&\bullet&\bullet&0&\bullet&0&\mathbf{1}&0&0\\
			0&0&0&0&0&0&0&0&0&\bullet&\bullet&0&\bullet&\mathbf{1}&0&0&0\\
			0&0&0&0&0&0&0&0&0&\bullet&\bullet&\mathbf{1}&0&0&0&0&0
		\end{array}\right].
		$$
		
		Denote
		$$
		\mathcal{F}_1=\begin{matrix}
			&\bullet&\bullet&\bullet&\bullet&\bullet&\bullet\\
			&&&\bullet&\bullet&\bullet&\bullet\\
			&&&\bullet&\bullet&\bullet&\bullet\\
			&&&&&\bullet&\bullet\\
			&&&&&\bullet&\bullet\\
			&&&&&\bullet&\bullet\\
		\end{matrix}.
		$$
		
		Notice that $\mathcal{F}_1=[1,1,3,3,6,6]$, then by Lemma 7 there exists an optimal $[\mathcal{F}_1, 8, 3]_q$-FDRMC. Therefore, the size of $\mathcal{D}_{{\bar{v}}_1^{\prime}}$ obtained from the bilateral multilevel construction is $q^{8}$.
	\end{proof}

	\begin{coro}\label{coro5}
		Let $n=18$, $k=7$, $\delta=3$. We have
		$
		\bar{A}_q(18,6,7) \geq q^{55}+q^{46}+q^{40}+(q^6+q^5+q^4+q^3+q^2+q+1)(q^4+q^3+q^2+q+1)(q^2-q+1)(q^{22}-(q^{11}-1)(q^3+q^2+q)-1)+q^{22}+q^{18}+q^{16}+q^{12}+q^{8}$.
	\end{coro}
	
	\begin{proof}
		We construct a set of identifying vectors as follows:
		\begin{align*}
			\mathcal{A}=\{v_1=(1111111 \underbrace{0 \cdots 0}_{11}), v_2=(1111000111 \underbrace{0 \cdots 0}_8), v_3=(1000111111 \underbrace{0 \cdots 0}_8)\},
		\end{align*}
		a set of inverse identifying vectors as follows:
		\begin{align*}
			\hat{\mathcal{A}}=\{(\hat{v}=(\underbrace{0 \cdots 0}_{11} 1111111)\},
		\end{align*}
		and two sets of bilateral identifying vectors as follows:
		\begin{align*}
			\mathcal{B}=\{\bar{v}_1=(1111000000  \bar{0} \bar{0} \hat{0} \hat{0}\hat{0} \hat{1} \hat{1} \hat{1}), \bar{v}_2=(1100110000 \bar{0}\bar{0} \hat{0} \hat{0} \hat{1} \hat{0} \hat{1} \hat{1}),
			\bar{v}_3=(0011110000\bar{0}\bar{0}\hat{0}\hat{0}\hat{1}\hat{1}\hat{0}\hat{1})\},
		\end{align*}
		\begin{align*}
			{\mathcal{B}_1}=\{\bar{v}_1^{\prime}=(0011001100 \bar{0}\bar{0} \hat{1} \hat{0} \hat{1} \hat{1} \hat{0} \hat{0}), {\bar{v}}_2^{\prime}=(0000111100  \bar{0}\bar{0} \hat{1} \hat{1} \hat{0} \hat{1} \hat{0} \hat{0})\}.
		\end{align*}
		The following proof is similar to Corollary 3 and we omit it here.
	\end{proof}
	
	\begin{rem}\label{rem2}
		Note that in \cite{19}, the authors calculated  $\bar{A}_q(18,6,7)\geq q^{55}+q^{46}+q^{40}+(q^6+q^5+q^4+q^3+q^2+q+1)(q^4+q^3+q^2+q+1)(q^2+q+1)(q^{22}-(q^{11}-1)(q^3+q^2+q)-1)+2(q^4+q^3+q^2+q+1)(q^2-q+1)(q^2+q+1)$, but there was an error in this result. The correct result should be $\bar{A}_q(18,6,7)\geq q^{55}+q^{46}+q^{40}+(q^6+q^5+q^4+q^3+q^2+q+1)(q^4+q^3+q^2+q+1)(q^2-q+1)(q^{22}-(q^{11}-1)(q^3+q^2+q)-1)+2(q^4+q^3+q^2+q+1)(q^2-q+1)(q^2+q+1)$. Clearly, the correct result is smaller than the result under our construction.
	\end{rem}
	
	\begin{coro}\label{coro6}
		Let $n=19$, $k=7$, $\delta=3$. Then we have
		$
		\bar{A}_q(19,6,7) \geq q^{60}+q^{51}+q^{45}+(q^6+q^5+q^4+q^3+q^2+q+1)(q^4+q^3+q^2+q+1)(q^2-q+1)(q^{22}-(q^{11}-1)(q^3+q^2+q)-1)+q^{27}+2q^{23}+q^{19}+q^{15}.$ 
	\end{coro}
	\begin{proof}
		We construct a set of identifying vectors as follows:
		\begin{align*}
			\mathcal{A}=\{v_1=(1111111 \underbrace{0 \cdots 0}_{12}), v_2=(1111000111 \underbrace{0 \cdots 0}_9), v_3=(1000111111 \underbrace{0 \cdots 0}_9)\},
		\end{align*}
		a set of inverse identifying vectors as follows:
		\begin{align*}
			\hat{\mathcal{A}}=\{(\hat{v}=(\underbrace{0 \cdots 0}_{12} 1111111)\},
		\end{align*}
		and two sets of bilateral identifying vectors as follows:
		\begin{align*}
			\mathcal{B}=\{\bar{v}_1=(1111000000 \bar{0} \bar{0} \bar{0} \hat{0} \hat{0} \hat{0} \hat{1} \hat{1} \hat{1}), \bar{v}_2=(1100110000 \bar{0}\bar{0}\bar{0} \hat{0} \hat{0} \hat{1} \hat{0} \hat{1} \hat{1}),
			\bar{v}_3=(0011110000\bar{0}\bar{0}\bar{0}\hat{0}\hat{0}\hat{1}\hat{1}\hat{0}\hat{1})\},
		\end{align*}
		\begin{align*}
			{\mathcal{B}_1}=\{{\bar{v}}_1^{\prime}=(0011001100 \bar{0}\bar{0}\bar{0} \hat{1} \hat{0} \hat{1} \hat{1} \hat{0} \hat{0}), {\bar{v}}_2^{\prime}=(0000111100 \bar{0} \bar{0}\bar{0} \hat{1} \hat{1} \hat{0} \hat{1} \hat{0} \hat{0})\}.
		\end{align*}
		The following proof is similar to Corollary 3 and we omit it here.
	\end{proof}
	
	\begin{rem}\label{rem3} Combining Corollary 6 and Lemma 1 with $n =19$, $\delta=3$, and $k=6$, we can calculate the ratio between the lower
bound and the upper bound of the CDCs: $\frac{the~lower~bound~of~\mathcal{C}}{the~upper~bound~of~\mathcal{C}} \geq 0.94548.$
\end{rem}
	\section{Conclusions}

	Subspace codes, especially CDCs, have received widespread attention due to their applications in random network coding. From a practical standpoint, at least for applications in network coding, many open questions (notably for the design)  arose in the framework of constant-dimension lifted rank-metric codes. This paper derived several constructions for CDCs. We presented the inverse bilateral multilevel construction by introducing inverse bilateral identifying vectors and inverse bilateral Ferrers diagram rank-metric codes. Furthermore, by introducing a new set of bilateral identifying vectors, we also provided another construction for CDCs. The inverse bilateral Ferrers diagram rank-metric code plays an important role in our construction and is meaningful to research. In this context, effectively identifying vectors is still an open and challenging task, but the multilevel inserting construction helps reduce the problem.
	For our constructions, we established lower bounds on the CDCs with the parameters $(n, d, k)\in \{(18, 6, 9), (17, 6, 6), (18, 6, 7), (19, 6, 7)\}$, which are larger than the known lower bounds of CDCs. Specific numerical results were given in Table 1 using the Software system MAGMA and the results were verified by MAGMA. In the future, we will consider constructing more optimal FDRMCs, which is beneficial for improving the size of CDCs with other (larger) parameters.
	
	\begin{table}[H]
\centering
\caption{New lower bounds of $\bar{A}_q(n,2\delta,k)$}
\vskip 2mm \setlength{\tabcolsep}{6pt}
{
\begin{tabular}{|c|c|c|c|>{\centering\arraybackslash}m{2cm}|}

\cline{1-5}
$\bar{A}_q(n,2\delta,k)$&  New lower bounds &Old lower bounds&Differences & References \\
\cline{1-5}
$\bar{A}_2(18,6,9)$&92715452252\underline{90474496}&92715452252\underline{88115199}&2359297&\\
\cline{1-4}
$\bar{A}_3(18,6,9)$&\makecell[c]{1144661280188263\\3236\underline{77419096134}}&\makecell[c]{1144661280188263\\3236\underline{66571322442}}&\makecell[c]{10847773\\692}&\\
\cline{1-4}
$\bar{A}_4(18,6,9)$&\makecell[c]{8507105814618280799\\887011\underline{9931236581376}}&\makecell[c]{8507105814618280799\\887011\underline{5464470593536}}&\makecell[c]{4466765\\987840}& \\
\cline{1-4}
$\bar{A}_5(18,6,9)$&\makecell[c]{1084202899657109779\\21242760305\underline{54916880}\\\underline{8593750}}&\makecell[c]{1084202899657109779\\21242760305\underline{06851695}\\\underline{3125000}}&\makecell[c]{4796518\\55468750}& \cite{30}\\
\cline{1-4}
$\bar{A}_7(18,6,9)$&\makecell[c]{1742515033889755513\\18887574330215518\underline{56}\\\underline{7696430750492486}}& \makecell[c]{1742515033889755513\\18887574330215518\underline{00}\\\underline{7522153069298030}}&\makecell[c]{492474277\\681194456}&\\
\cline{1-4}
$\bar{A}_8(18,6,9)$&\makecell[c]{7846377237219197911\\383819591537893969\underline{8}\\\underline{8803423687298514944}}&\makecell[c]{7846377237219197911\\383819591537893969\underline{7}\\\underline{9562037251934257152}}&\makecell[c]{92413864\\35364257\\792}&\\
\cline{1-4}
$\bar{A}_9(18,6,9)$&\makecell[c]{13100205124938663392\\06870324857588258367\\\underline{421841432175735022246}}&\makecell[c]{13100205124938663392\\06870324857588258367\\\underline{312272348408925663916}}&\makecell[c]{10956908\\37668093\\58330}&\\
\cline{1-5}
					
$\bar{A}_3(17,6,6)$&98482278675490\underline{6111880}&98482278675490\underline{0790910}&5320970&\\
\cline{1-4}
$\bar{A}_4(17,6,6)$&\makecell[c]{309486208859711440\underline{565}\\\underline{256219}}&\makecell[c]{309486208859711440\underline{279}\\\underline{978012}}&285278207&\\
\cline{1-4}
$\bar{A}_5(17,6,6)$&\makecell[c]{56843448197469116506\underline{6}\\\underline{2156035194}}&\makecell[c]{56843448197469116506\underline{5}\\\underline{5807988320}}&\makecell[c]{6348046\\874}&\\
\cline{1-4}
$\bar{A}_7(17,6,6)$&\makecell[c]{152867010118656961337\\69150\underline{856284558796}}& \makecell[c]{152867010118656961337\\69150\underline{164214433946}}&\makecell[c]{69207012\\4850}& \cite{18}\\
\cline{1-4}
$\bar{A}_8(17,6,6)$&\makecell[c]{544451791137906278523\\29396\underline{94101114276215} }&\makecell[c]{ 544451791137906278523\\29396\underline{89634331511160}}&\makecell[c]{44667827\\65055}&\\
\cline{1-4}
$\bar{A}_9(17,6,6)$&\makecell[c]{969773732294112791690\\3699209\underline{76223779322734}} & \makecell[c]{969773732294112791690\\3699209\underline{53064514284572}}&\makecell[c]{23159265\\038162}&\\
\cline{1-5}
$\bar{A}_3(18,6,7)$&\makecell[c]{174458086133950\underline{601507}\\\underline{064752}}& \makecell[c]{174458086133950\underline{569695}\\\underline{021953}}&\makecell[c]{3181204\\2799}&\\
\cline{1-4}
$\bar{A}_4(18,6,7)$&\makecell[c]{12980791676032157421\\\underline{80577631319615}}& \makecell[c]{ 12980791676032157421\\\underline{62912414174601}}&\makecell[c]{17665217\\145014}&\\
\cline{1-4}
$\bar{A}_5(18,6,7)$&\makecell[c]{277555898273931997862\\55\underline{3960563246907624}}& \makecell[c]{ 277555898273931997862\\55\underline{1572409927221361}}&\makecell[c]{23881533\\19686263}& \cite{19} and\\
\cline{1-4}
$\bar{A}_7(18,6,7)$&\makecell[c]{302268027208297537838\\299541\underline{73786880863008}\\\underline{449184}}&  \makecell[c]{302268027208297537838\\299541\underline{69875398154063}\\\underline{659585}}&\makecell[c]{39114827\\08944789\\599}& Remark 2 \\
\cline{1-4}
$\bar{A}_8(18,6,7)$&	\makecell[c]{467680527430393663434\\529587833\underline{788926286522}\\\underline{91957247}}& \makecell[c]{ 467680527430393663434\\529587833\underline{050873564152}\\\underline{70259473}}&\makecell[c] {73805272\\23702170\\2774}&\\
\cline{1-4}
$\bar{A}_9(18,6,7)$&	\makecell[c]{30432527300256357008\\3080541775\underline{904303773}\\\underline{50854469863880}}& \makecell[c]{30432527300256357008\\3080541775\underline{894454545}\\\underline{00732998260801}}&\makecell[c]{98492285\\01214706\\03079}&\\
\cline{1-5}
$\bar{A}_3(19,6,7)$&	\makecell[c]{4239331492375343932\\\underline{46}\underline{93652326}} & \makecell[c]{4239331492375343932\\\underline{35}\underline{17041952}}&\makecell[c]{1176610\\374}&\\
\cline{1-4}
$\bar{A}_4(19,6,7)$&\makecell[c]{13292330676252636392\\0961\underline{2245368813119}}& \makecell[c]{ 13292330676252636392\\0961\underline{1969417164351}}&\makecell[c]{27595164\\8768}&\\
\cline{1-5}
\end{tabular}}
\end{table}

	\begin{table}[H]
	\centering
	
	\vskip 2mm \setlength{\tabcolsep}{6pt}
{
\begin{tabular}{|c|c|c|c|>{\centering\arraybackslash}m{2cm}|}

\cline{1-5}

$\bar{A}_5(19,6,7)$	&\makecell[c]{867362182106035125792\\0103091\underline{3141407659}\\\underline{5124}}&\makecell[c]{ 867362182106035125792\\0103091\underline{1231007268}\\\underline{8874}}&\makecell[c]{19104003\\906250}&    \\
\cline{1-4}
$\bar{A}_7(19,6,7)$&\makecell[c]{5080218733289856707\\616639748793042\underline{5524}\\\underline{4135597816466}}&\makecell[c]{5080218733289856707\\616639748793042\underline{4384}\\	\underline{0492850933380}}&\makecell[c]{114036427\\46883086}& \cite{30}\\
\cline{1-4}
$\bar{A}_8(19,6,7)$&	\makecell[c]{1532495552283913956\\149369445213372457\\\underline{846736476350095871}}&\makecell[c]{ 1532495552283913956\\149369445213372457\\\underline{702586103902151167}}&\makecell[c]{144150372\\447944704}&\\
\cline{1-4}
$\bar{A}_9(19,6,7)$&	\makecell[c]{1797010304552837624\\96487559256535156470\\\underline{8692531023075402144}}&\makecell[c]{ 1797010304552837624\\96487559256535156470\\\underline{7341473414270315406}}&\makecell[c]{13510576\\08805086\\738}&\\
\hline
\end{tabular}}
\end{table}

\section*{Acknowledgement}
G. Wang and X. Gao are supported by the National Natural Science Foundation of China (No. 12301670), the Natural Science Foundation of Tianjin (No. 23JCQNJC00050), the Scientific Research Project of Tianjin Education Commission (No. 2022KJ075), the Fundamental Research Funds for the Central Universities of China (No. 3122023QD25, 3122024PT24) and the Graduate Student Research and Innovation Fund of Civil Aviation University of China. The French Agence Nationale de la Recherche partially supported the third author's work through ANR BARRACUDA (ANR-21-CE39-0009). F.-W Fu is supported by the National Key Research and Development Program of China (Grant No. 2018YFA0704703), the National Natural Science Foundation of China (Grant No. 61971243), the Natural Science Foundation of Tianjin (20JCZDJC00610), the Fundamental Research Funds for the Central Universities of China (Nankai University), and the Nankai Zhide Foundation.


		

\begin{thebibliography}{99}
				
				
\bibitem{1} J. Antrobus and H. Gluesing-Luerssen, Maximal Ferrers diagram codes: constructions and genericity considerations, IEEE Trans. Inf. Theory, 65 (2019), 6204-6223.
				
				
\bibitem{2}H. Chen, X. He, J. Weng, and L. Xu, New constructions of subspace codes using subsets of MRD codes in several blocks, IEEE Trans. Inf. Theory, 66 (2020), 5317-5321.
	
\bibitem{cos}A. Cossidente, S. Kurz, G. Marino and F. Pavese, Combining subspace codes, Adv. Math.Comm., 17 (2023), 536-550.			

\bibitem{3}P. Delsarte, Bilinear forms over a finite field, with applications to coding theory, J. Combin. Theory A, 25 (1978), 226-241.
				
 
				
\bibitem{4}T. Etzion and N. Silberstein, Error-correcting codes in projective spaces via rank-metric codes and Ferrers diagrams, IEEE Trans. Inf. Theory, 55 (2009), 2909-2919.
				
\bibitem{5}T. Etzion, E. Gorla, A.Ravagnani and A. Wachter-Zeh, Optimal Ferrers diagram rank-metric codes, IEEE Trans. Inf. Theory, 62 (2016), 1616-1630.
				
\bibitem{6} T. Etzion and N. Silberstein, Codes and designs related to lifted MRD codes, IEEE Trans. Inf. Theory, 59 (2013), 1004-1017.
\bibitem{7}T. Etzion and A. Vardy, Error-correcting codes in projective spaces, IEEE Trans. Inf. Theory, 57 (2011), 1165-1173.
\bibitem{8} E. Gorla and A. Ravagnani, Subspace codes from Ferrers diagrams, J. Algebra Appl., 16 (2017), 1750131.
				
\bibitem{9} \`{E}.M. Gabidulin, Theory of codes with maximum rank distance, Problems Inf. Transmiss., 21 (1985), 3-16.
				
\bibitem{33}M. Gadouleau and Z. Yan, Constant-rank codes and their connection to constant-dimension codes, IEEE Trans. Inf. Theory, 56 (2010), 3207-3216. 
				
\bibitem{10}H. Gluesing-Luerssen and C. Troha, Construction of subspace codes through linkage, Adv. Math. Commun., 10 (2016), 525-540.
				
				
\bibitem{11}D. Heinlein, New LMRD code bounds for constant dimension codes and improved constructions, IEEE Trans. Inf. Theory, 65 (2019), 4822-4830.
				
\bibitem{12}D. Heinlein, Generalized linkage construction for constant-dimension codes, IEEE Trans. Inf. Theory, 67 (2021), 705-715.
				
				
\bibitem{13}D. Heinlein and S. Kurz, Coset construction for subspace codes, IEEE Trans. Inf. Theory, 63 (2017), 7651-7660.
\bibitem{14}D. Heinlein, M. Kiermaier, S. Kurz, and A. Wassermann, Tables of
subspace codes, http://subspacecodes.uni-bayreuth.de.
				
\bibitem{15}X. He, Y. Chen, Z. Zhang and K. Zhou, Parallel sub-code construction for constant-dimension codes, Des. Codes Cryptogr., 90 (2022), 2991-3001.
				
				
\bibitem{16} R. K\"{o}tter and F.R. Kschischang, Coding for errors and erasures in random network coding, IEEE Trans. Inf. Theory, 54 (2008), 3579-3591.
				
\bibitem{17}M. Kiermaier and S. Kurz, On the lengths of divisible codes, IEEE Trans. Inf. Theory, 66 (2020), 4051-4060.
				
				
\bibitem{18}  S. Liu, Y. Chang and T. Feng, Parallel multilevel constructions for constant dimension codes, IEEE Trans. Inf. Theory, 66 (2020), 6884-6897.
				
				
\bibitem{19} S. Liu and L. Ji, Double multilevel constructions for constant dimension codes, IEEE Trans. Inf. Theory, 69 (2023), 157-168.
				
\bibitem{20} S. Liu, Y. Chang and T. Feng, Several classes of optimal Ferrers diagram rank-metric codes, Linear Algebra Appl., 581 (2019), 128-144.
				
				
\bibitem{21}S. Liu, Y. Chang and T. Feng, Constructions for optimal Ferrers diagram rank-metric codes, IEEE Trans. Inf. Theory, 65 (2019), 4115-4130.


				
\bibitem{22}Y. Niu, Q. Yue and D. Huang, New constant dimension subspace codes from generalized inserting construction, IEEE Commun. Lett., 25 (2021), 1066-1069.
				
				
\bibitem{23} Y. Niu, Q. Yue and D. Huang, New constant dimension subspace codes from parallel linkage construction and multilevel construction, Cryptogr. Commun., 14 (2022), 201-214.
				
\bibitem{24} R. M. Roth, Maximum-rank array codes and their application to crisscross error correction, IEEE Trans. Inf. Theory, 37 (1991), 328-336.
				
				
				
\bibitem{25}N. Silberstein and T. Etzion, Large constant dimension codes and lexicodes, Adv. Math. Commun., 5 (2011), 177-189.
				
				
\bibitem{26}D. Silva, F.R. Kschischang and R. K\"{o}tter, A rank-metric approach to error control in random network coding, IEEE Trans. Inf. Theory, 54 (2008), 3951-3967.
				
\bibitem{27}N. Silberstein and A.-L. Trautmann, Subspace codes based on graph matchings, Ferrers diagrams, and pending blocks, IEEE Trans. Inf. Theory, 61 (2015), 3937-3953.
				
\bibitem{28}A.-L. Trautmann and J. Rosenthal, New improvements on the Echelon Ferrers construction, in Proc. 19th Int. Symp. Math. Theory Netw. Syst., Jul. (2010), 405-408.
				
\bibitem{29}L. Xu and H. Chen, New constant-dimension subspace codes from maximum rank distance codes, IEEE Trans. Inf. Theory, 64 (2018), 6315-6319.
				
\bibitem{30}S. Yu, L. Ji and S. Liu, Bilateral multilevel construction of constant dimension codes, Adv. Math. Commun., 16 (2022), 1165-1183.
				
\bibitem{31} T. Zhang and G. Ge, Constructions of optimal Ferrers diagram rank metric codes, Des. Codes Cryptogr., 87 (2019), 107-121.
				
				
				
\end{thebibliography}
\end{document}